\documentclass[10pt,journal,compsoc]{IEEEtran}

\ifCLASSOPTIONcompsoc
  \usepackage[nocompress]{cite}
\else
  \usepackage{cite}
\fi

\makeatletter
\def\ps@headings{%
\def\@oddhead{\mbox{}\scriptsize\rightmark \hfil \thepage}%
\def\@evenhead{\scriptsize\thepage \hfil \leftmark\mbox{}}%
\def\@oddfoot{}%
\def\@evenfoot{}}
\makeatother 

\usepackage{tipa}
\usepackage{stmaryrd}
\usepackage{pifont}
\usepackage{amssymb}
\usepackage{amsmath,amsthm}
\usepackage{cite}
\usepackage{amsfonts}
\usepackage{txfonts}
\usepackage{amsmath}
\usepackage{latexsym}
\usepackage{cases}
\usepackage{subfigure}
\usepackage{algorithm}
\usepackage{algorithmic}
\usepackage{bbding}
\usepackage{times}
\usepackage{array}
\usepackage{color}
\usepackage{bm}
\usepackage{mathtools}
\usepackage{setspace}
\usepackage{color}

\ifCLASSINFOpdf
\else
\fi

\begin{document}

\title{\huge{MSCET: A Multi-Scenario Offloading Schedule for Biomedical Data Processing and Analysis in Cloud-Edge-Terminal Collaborative Vehicular Networks}}

\author{Zhichen~Ni,~Honglong~Chen,~{\em Senior~Member,~IEEE},~Zhe~Li,~Xiaomeng~Wang,\\~Na~Yan,~Weifeng~Liu,~{\em Senior~Member,~IEEE},~and~Feng~Xia,~{\em Senior~Member,~IEEE}
\IEEEcompsocitemizethanks{\IEEEcompsocthanksitem Zhichen Ni, Honglong Chen, Zhe Li, Xiaomeng Wang, Na Yan and Weifeng Liu are with College of Control Science and Engineering,
China University of Petroleum (East China), Qingdao 266580, P. R. China.\protect
\IEEEcompsocthanksitem Feng Xia is with the School of Science, Engineering and Information Technology, Federation University Australia, Ballarat, VIC 3353, Australia.}
\thanks{$\#$ Corresponding author: Honglong Chen. Email: chenhl@upc.edu.cn.}}

\IEEEtitleabstractindextext{%
\begin{abstract}
  With the rapid development of Artificial Intelligence (AI) and Internet of Things (IoTs), an increasing number of computation intensive or delay sensitive biomedical data processing and analysis tasks are produced in vehicles, bringing more and more challenges to the biometric monitoring of drivers. Edge computing is a new paradigm to solve these challenges by offloading tasks from the resource-limited vehicles to Edge Servers (ESs) in Road Side Units (RSUs). However, most of the traditional offloading schedules for vehicular networks concentrate on the edge, while some tasks may be too complex for ESs to process. To this end, we consider a collaborative vehicular network in which the cloud, edge and terminal can cooperate with each other to accomplish the tasks. The vehicles can offload the computation intensive tasks to the cloud to save the resource of edge. We further construct the virtual resource pool which can integrate the resource of multiple ESs since some regions may be covered by multiple RSUs. In this paper, we propose a {M}ulti-{S}cenario offloading schedule for biomedical data processing and analysis in {C}loud-{E}dge-{T}erminal collaborative vehicular networks called MSCET. The parameters of the proposed MSCET are optimized to maximize the system utility. We also conduct extensive simulations to evaluate the proposed MSCET and the results illustrate that MSCET outperforms other existing schedules.
\end{abstract}

}

\maketitle

\IEEEdisplaynontitleabstractindextext
\IEEEpeerreviewmaketitle

\ifCLASSOPTIONcompsoc
\IEEEraisesectionheading{\section{Introduction}\label{sec:introduction}}
\else
\section{Introduction}
\label{sec:introduction}
\fi

\IEEEPARstart{T}{he} rapid development of Artificial Intelligence (AI)~\cite{Chen:2021,Li:2021} and Internet of Things (IoTs)~\cite{Aixin:2020,Chen:20182,Lin:2021} has paved a way to realize modern applications in vehicles. For example, on-board AI plays an important role in the analysis of driver biomedical data~\cite{Feng:2018,Qiu:2019}. However, these applications are mostly computation intensive or delay sensitive while the computation resource of vehicles is relatively limited. This conflict may decline the Quality of Service (QoS) and eventually restrict the development of vehicular networks~\cite{Mao:2019,Song:2017}. Cloud computing can be used in such kind of scenario~\cite{Seabolt:2020,Huang:2021}, in which vehicles can offload tasks to resource-rich cloud servers via wireless network. However, the latency for the cloud computing may be too long, resulting in that some tasks cannot be accomplished in time.

Edge computing is a new paradigm and has attracted lots of attentions from both the research and industrial community. In contrast to cloud computing, edge computing brings computation services from the remote cloud to the network edge, enabling that some tasks can be quickly processed near the vehicles~\cite{Khan:2019}. Recently, edge computing for vehicular networks has been studied, and some research works~\cite{Zhang:2019,Lingjun:2019} have been proposed.

Most of the traditional vehicular networks concentrate on Edge Servers (ESs) in Road Side Units (RSUs), while some tasks may be too complex for ESs to process. A potential solution is to combine cloud computing with edge computing, where the vehicles can offload the complex and delay-tolerant tasks to Cloud Servers (CSs)~\cite{Kai:2021}. Moreover, with the development of Network Function Virtualization (NFV) technology, the construction of virtual resource pool becomes feasible~\cite{Li:2019}. Some regions may be covered by multiple RSUs in the vehicular networks~\cite{Wu:2021}. In such situation, the virtual resource pool can be utilized in the offloading schedule to integrate the resource of multiple ESs. And we further adopt Software Defined Network (SDN) to manage the virtual resource pool to improve its performance.

In this paper, we concentrate on investigating the problem of efficient and reliable multi-scenario offloading schedule for biomedical data processing and analysis in Cloud-Edge-Terminal collaborative vehicular networks. To this end, there are two challenges that need to be well addressed. The first one is the selection of vehicles' offloading targets, which can affect the optimization of the QoS utility. For instance, it will cause the overload of RSU if a numerous number of vehicles offload tasks to the same RSU, which may result in the timeout of tasks~\cite{Feng:2015}. Thus, we adopt a matching method for the selection of each vehicle's offloading target and set a dynamic maximum connection limit for each RSU. The second challenge is how to obtain the reasonable offloading ratio and resource allocation to improve the system utility. To address these challenges, we consider a Cloud-Edge-Terminal collaborative vehicular network consisting of a CS, a set of RSUs, each of which is equipped with an ES, and a set of vehicles. Each vehicle has a computation task and can select servers to offload the task to. By jointly optimizing the offloading targets selection, offloading ratio and resource allocation, we obtain the near-optimal schedule to maximize the system utility.

{\bf The main contributions of this paper} are the following:

\begin{itemize}

\item We model the offloading decision and resource allocation for Cloud-Edge-Terminal collaborative vehicular networks, and consider both the profits and QoS in the system utility;

\item To improve the processing efficiency of tasks, we construct the virtual resource pool to integrate the resource of multiple ESs and further introduce the SDN for the centralized management of the pool.

\item We propose a Multi-Scenario offloading schedule for biomedical data processing and analysis tasks in Cloud-Edge-Terminal collaborative vehicular networks (MSCET), which optimizes the resource allocation, offloading targets and offloading ratio to maximize the system utility;

\item We conduct extensive simulations to evaluate the performance of the proposed MSCET and the results illustrate that MSCET outperforms other existing schedules.

\end{itemize}

The rest of this paper is organized as follows. Section~\ref{sec2} reviews the related work. Section~\ref{sec3} presents the system model and problem formulation. We propose the MSCET schedule with theoretical analysis in Sections~\ref{sec4}. In Section~\ref{sec5}, we conduct the performance evaluation. Finally, this paper is concluded in Section~\ref{sec6}.

\section{Related Work}\label{sec2}

Edge computing has been well studied in recent years and many research works have been proposed. Typically, some of them focused on how to make the reasonable offloading selections of servers~\cite{Guo2:2018,Zhao:2017}, while some others showed interest in optimizing the offloading ratio and resource allocation to maximize the system utility~\cite{Yu:2020,You:2017}. Guo~\textit{et al.}~\cite{Guo2:2018} proposed a heuristic greedy offloading scheme for the mobile edge computation offloading problem in ultra-dense network and the schedule showed superior performance by conducting computation offloading over multiple mobile edge computing servers. For optimal service provisioning, Yu~\textit{et al.}~\cite{Yu:2020} formulated an optimization problem to minimize the weighted sum of the service delay of all IoT devices and energy consumption by jointly optimizing the UAV position, communication and computing resource allocation, and task splitting decisions. Then they developed an efficient algorithm based on the successive convex approximation to obtain suboptimal solutions.

In order to address the conflict between the computation intensive or delay sensitive applications and resources-constrained vehicles, there are some research works on the combination of edge computing and vehicular networks in recent years~\cite{Dai:2019,Sorkhoh:2019,Liu:2018}. Dai~\textit{et al.}~\cite{Dai:2019} studied the resource allocation problem for a multi-user multi-server vehicular edge computing system and proposed a low-complexity algorithm to jointly optimize the server selection, offloading ratio and computation resource. The algorithm showed the superior performance in reducing the task processing delay while avoiding server overload. Liu~\textit{et al.}~\cite{Liu:2018} considered the multiple vehicles' computation offloading problem in vehicular edge networks and formulated it as a multi-user computation offloading game problem. Then a distributed algorithm based on game theory was proposed to reduce the latency of the computation offloading of vehicles.

Some research works~\cite{Zhang:2020,Liu2:2017} have studied on the integration of the SDN and vehicular edge computing networks. An SDN-enabled network architecture that integrated different types of access technologies was constructed in~\cite{Liu2:2017} to provide low-latency and high-reliability communication in mobile edge computing networks. Extensive numerical results showed that their proposed architecture has a significant improvement in the quality of user experience. Aiming to minimize the processing delay of the computation task of vehicles, Zhang~\textit{et al.}~\cite{Zhang:2020} proposed an SDN based load-balancing task offloading schedule in fiber-wireless (FiWi) enhanced vehicular edge computing networks, where SDN was introduced to provide supports for the centralized network and vehicle information management.

The delay of task processing can be decreased greatly with the assistance of edge computing, while some tasks may be too complex for edge servers to process. A potential solution is to combine the cloud computing with edge computing. Some works have focused on the cooperation between cloud and edge in recent years~\cite{Shen:2020,Zhao:2019}. Zhao~\textit{et al.}~\cite{Zhao:2019} presented a collaborative approach based on mobile edge computing and cloud computing in vehicular networks and further proposed a collaborative computation offloading and resource allocation optimization schedule to improve the system utility. To tackle the Problem of edge server overload, Shen~\textit{et al.}~\cite{Shen:2020} proposed a dynamic task offloading method with minority game in cloud-edge computing. However, to simplify the model, they assumed that the tasks are offloaded to either the cloud server or the edge servers, which degraded the performance of cloud-edge cooperation networks.

The above works can improve the ability of the system for completing the complex tasks by utilizing the cooperation of cloud computing and edge computing. However, some of them only consider a relatively simple case such as binary decision. Moreover, they mainly focus on a single scenario and do not concentrate on the influence of the complex communication environment in vehicular networks. In this paper, we propose a multi-scenario offloading schedule for Cloud-Edge-Terminal cooperation vehicular networks such that each vehicle can adaptively select its schedule based on their environments.

\section{System Model and Problem Formulation}\label{sec3}

\subsection{Network Framework}

As shown in Fig.~\ref{system}, we consider a Cloud-Edge-Terminal collaborative vehicular network consisting of a Cloud Server (CS), a set of Road Side Units (RSUs) $\mathcal M =\lbrace E_1,E_2,\cdots,E_M \rbrace$, and a set of vehicles $\mathcal N =\lbrace V_1,V_2,\cdots,V_N \rbrace$ that represent the terminal users. RSUs are randomly deployed in the area and each RSU is equipped with a resource limited Edge Server (ES). The coverage size of each RSU may be different due to the complex communication environment. Let $r = \lbrace r_1,r_2,\cdots,r_M \rbrace$ denote the coverage radius of each RSU. Some {regions} may be covered by multiple RSUs as the RSUs are randomly distributed, and we define them as the overlapping regions, such as region A shown in Fig.~\ref{system}, and other regions as the general regions. We use $l_i = \lbrace 0,1\rbrace$ to represent the region that vehicle $V_i$ is located in, where $l_i = 1$ represents that vehicle $V_i$ is located in overlapping region, and $l_i = 0$ represents that vehicle $V_i$ is located in general region.

In general region, we simplify the model to consider a unidirectional road, where $M_G$ RSUs are deployed along the road, as shown in Fig.~\ref{system2}. While in overlapping region, we construct the virtual resource pool to integrate computation resource of the multiple distributed ESs using Network Function Virtualization (NFV) technology to improve the processing efficiency of tasks, and Software Defined Network (SDN) is adopted to provide centralized management for the resource pool. There are three parts in SDN architecture, i.e., control plane, data plane and application plane, the main components and responsibilities of which are as follows:

\begin{itemize}

\item Control Plane: It is made up of an SDN controller and RSU controllers, where SDN controller provides global management based on the information of RSU controllers, and RSU controllers are responsible for information collection on the resource of ESs and the tasks of vehicles. Moreover, they can also forward the offloading schedule of SDN controller.

\item Data Plane: Data plane is a physical hardware layer that focuses on forwarding the data from RSUs and vehicles. It can be managed by either SDN controller or RSU controllers.

\item Application Plane: It contains a series of hypervisors. In our constructed virtual resource pool, the main hypervisors consist of resource migration, OpenFlow management and offloading decision.

\end{itemize}

\begin{figure}
  \centering
  \includegraphics[width=3in]{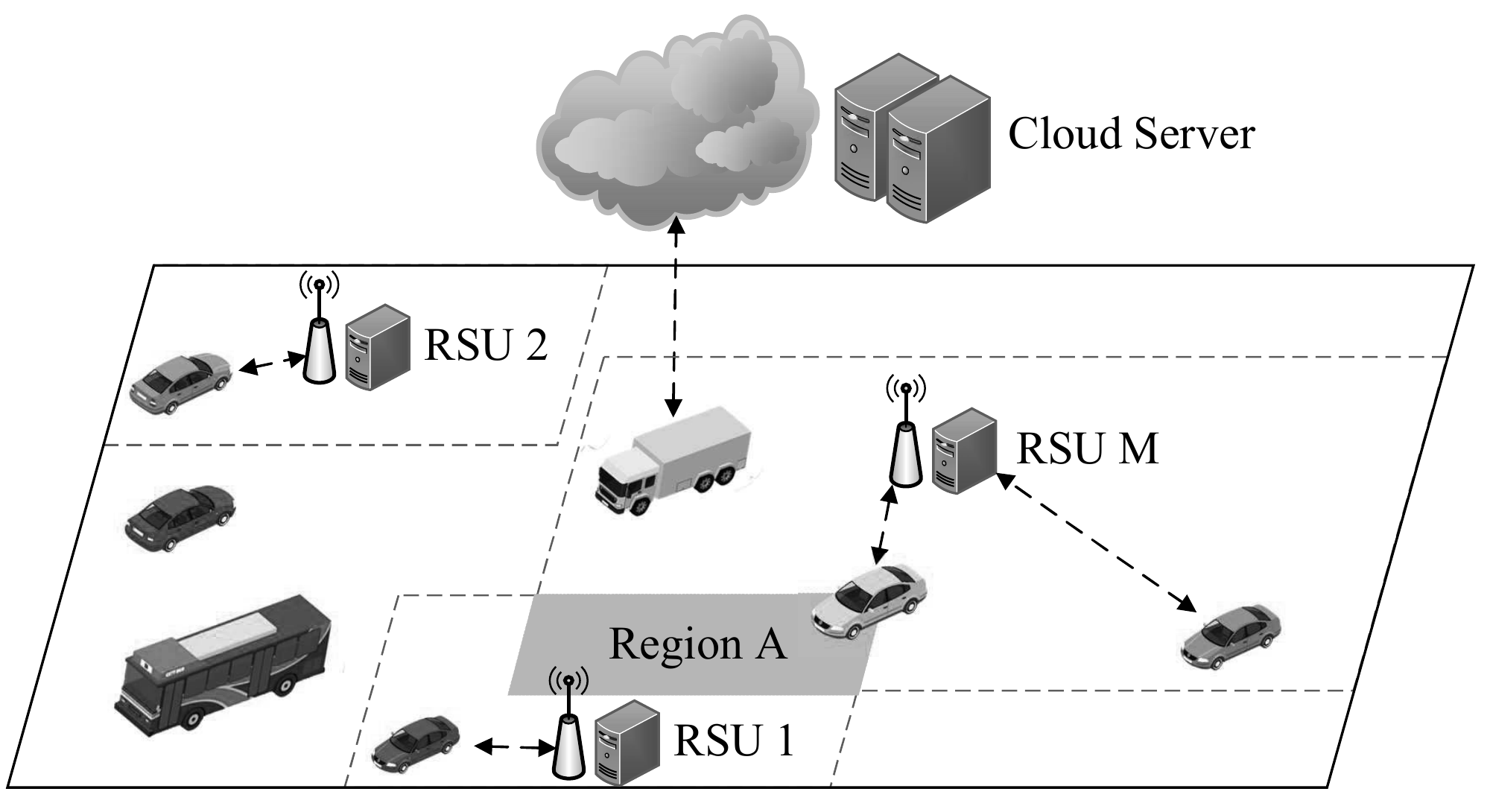}
  \caption{Task offloading in Cloud-Edge-Terminal collaborative vehicular networks.}\label{system}~\\
  \includegraphics[width=3in]{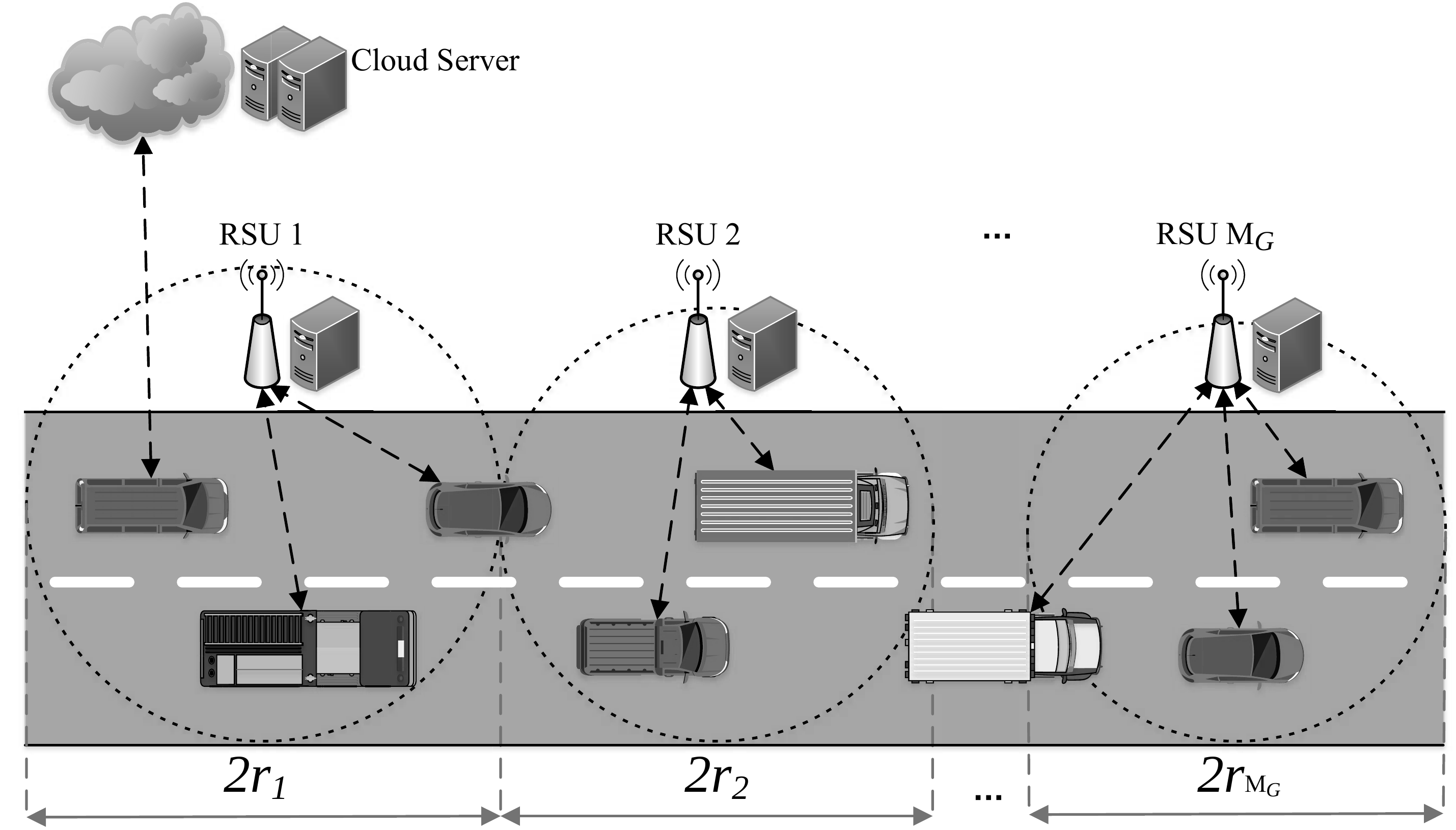}
  \caption{Task offloading in general region.}\label{system2}
\end{figure}

\subsection{Task Model}

Each vehicle $V_i$ is running on the road at a speed of $v_i$ and the safe distance between two vehicles is $L$. Vehicle $V_i$ has a computation intensive or delay sensitive task that contains three characteristics $\lbrace D_i,R_i,T_i^{max} \rbrace$, where $D_i$ {(in bits)} represents the size of task input data, $R_i$ {(in CPU cycles/bit) is the number of CPU cycles required to process 1-bit of the task}, and $T_i^{max}$ denotes the maximum tolerant delay to accomplish the task. Each task can be divided into three parts, which can be processed in the terminal, ES and CS in parallel. Let $\alpha_i = \lbrace \alpha_{i,e},\alpha_{i,c} \rbrace$  $(0 \leq \alpha_{i,e} + \alpha_{i,c} \leq 1)$ denote the offloading ratio of vehicle $V_i$, then $V_i$ offloads $\alpha_{i,c}D_i$ part of its task to the CS, offloads $\alpha_{i,e}D_i$ part of its task to the selected ES and processes the rest of its task locally. We denote $S_{i,j} = \lbrace 0,1\rbrace$ as the offloading decision of vehicle $V_i$ on RSU, where $S_{i,j} = 1$ if vehicle $V_i$ selects RSU $E_j$ as its offloading target, and $S_{i,j} = 0$ otherwise.

\subsection{Communication Model}

 There are two communication scenarios in our network, i.e., cloud communication and edge communication. For cloud communication, we assume that the data transmission rate is constant because of the long-distance between the vehicles and cloud~\cite{Sorkhoh:2019}. Let $\rho_{i,CS}$ denote the data transmission rate between vehicle $V_i$ and CS. The communication delay for cloud computing between vehicle $V_i$ and CS, denoted as $T_{i,CS}^{comm}$, can be formulated as:
\begin{equation}\label{transmission time}
\ T_{i,CS}^{comm} =  \frac{{\alpha_{i,c}}D_i}{\rho_{i,CS}} \, .
\end{equation}

For edge computing in general region, each vehicle can establish a communication link with one of its nearby RSUs using vehicle to infrastructure (V2I) technology. Assume that the wireless channel is stable within a time slot (e.g., a few seconds), while it may change in different slots because of the mobility of vehicles~\cite{Zeng:2017}. According to Shannon Theory, the data transmission rate between vehicle $V_i$ and RSU $E_j$, denoted as $\rho_{i,j}$, can be expressed as:
\begin{equation}\label{data transmission rate}
\ \rho_{i,j} = B\log \left (1 + \frac{P_iG_{i,j}}{\sigma + \sum_{i' = 1 , i' \neq i}^N P_{i'}G_{{i'},j} } \right) ,
\end{equation}
where $\sigma$ is white noise of the channel. $B$, $P_i$ and $G_{i,j}$ are the bandwidth, the transmission power of vehicle $V_i$ and the channel gain between vehicle $V_i$ and RSU $E_j$ respectively. Based on above analysis, the communication delay between vehicle $V_i$ and RSU $E_j$ for edge computing in general region, denoted as $T_{i,j}^{comm}$, is:
\begin{equation}\label{transmission time in GR}
\ T_{i,j}^{comm} = \frac {{\alpha_{i,e}}D_i}{\rho_{i,j}}.
\end{equation}

For edge computing in overlapping region, we denote $\rho_{i,P}$ as the data transmission rate between vehicle $V_i$ and the virtual resource pool, and $T_{i,P}^{comm}$ as the communication delay between vehicle $V_i$ and the virtual resource pool. The expressions of $\rho_{i,P}$ and $T_{i,P}^{comm}$ are similar to that of $\rho_{i,j}$ and $T_{i,j}^{comm}$, so we do not repeat them again.

By following the works in~\cite{Huang:2012,Xu:2016}, we ignore the download communication delay as the size of processed data is generally much smaller than that of input data when there is no server overload.

\vspace{-3mm}
\subsection{Computation Model}

As vehicle $V_i$ processes $(1 - \alpha_{i,e} - \alpha_{i,c})D_i$ part of its task locally and offloads the rest to the selected servers (i.e., ES and CS), we consider three different cases of the computation delay for ease of explanation.

\subsubsection{Local Computing}
Vehicle $V_i$ processes $(1 - \alpha_{i,e} - \alpha_{i,c})D_i$ part of its task locally using its own resource. Different vehicles may have different resource, and the amount of resource can be obtained by offline measurement method~\cite{Miettinen:2010}. Let $f_{i,V}$ (in CPU cycles/s) denote the computation resource of vehicle $V_i$, thus the local computation delay of $V_i$, denoted as $T_i^{local}$, can be expressed as:
\begin{equation}\label{Local computing}
\ T_i^{local} = \frac {{\Big( 1 - \alpha_{i,e} - \alpha_{i,c} \Big)}{D_iR_i}}{f_{i,V}}  \, .
\end{equation}

\subsubsection{Edge Computing}
{For edge computing in general region, vehicle $V_i$ offloads $\alpha_{i,e}D_i$ part of its task to RSU $E_j$.} As each ES is resource limited, it is necessary to allocate its computation resource. We denote $f_{i,j}$ (in CPU cycles/s) as the allocated resource of RSU $E_j$'s ES to vehicle $V_i$. The computation delay for edge computing in general region when vehicle $V_i$ offloads task to RSU $E_j$, denoted as $T_{i,j}^{comp} $, can be formulated as:
\begin{align}
\ &T_{i,j}^{comp} = \frac {{\alpha_{i,e}}D_iR_i}{f_{i,j}} \, , \label{Edge computing in GR} \\
&s.t. \quad 0 \leq \sum_{i = 1}^N S_{i,j} f_{i,j} \leq F_j \quad \forall j \in [1,\cdots,M_G] \tag{\ref{Edge computing in GR}{{a}}}\label{CR},
\end{align}
where $F_j$ is the total resource of ES in RSU $E_j$.

For edge computing in overlapping region, vehicles do not need to select the offloading target as the ESs of all the RSUs are integrated as the virtual resource pool. We denote $f_{i,P}$ (in CPU cycles/s) as the allocated resource from resource pool to vehicle $V_i$. The computation delay for edge computing in overlapping region of vehicle $V_i$, denoted as $T_{i,P}^{comp}$, is:
\begin{align}
\ & T_{i,P}^{comp} = \frac {{\alpha_{i,e}}D_iR_i}{f_{i,P}} \, , \label{Edge computing in OR} \\
&s.t. \quad 0 \leq \sum_{i = 1}^N f_{i,P} \leq F_{P} \, \tag{\ref{Edge computing in OR}{{a}}}\label{CP},
\end{align}
where $F_{P}$ is the total resource of the virtual resource pool.

\subsubsection{Cloud Computing}

As the CS is resource-rich, it does not need to consider the resource allocation. We denote $f_{i,CS}$ (in CPU cycles/s) as the resource CS provides for vehicle $V_i$, then the cloud computation delay of vehicle $V_i$, denoted as $T_{i,CS}^{comp}$, is:
\begin{equation}\label{Cloud computing}
\ T_{i,CS}^{comp} = \frac {{\alpha_{i,c}}D_iR_i}{f_{i,CS}} \, .
\end{equation}

\subsection{Processing Delay}

As mentioned above, each task is divided into three parts, which will be processed in the terminal, ES and CS in parallel. Thus, the task processing delay is dominated by the largest one of the local computing delay, ES processing delay and CS processing delay. Note that vehicle $V_i$ will establish a communication link with CS only when it finishes the communication with RSU, which means that the processing delay of CS consists of three parts, i.e., the computation delay, the communication delay between vehicle and RSU and the communication delay between vehicle and CS. Based on above analysis, the processing delay of vehicle $V_i$, denoted as $T_i$, can be expressed as:
\begin{equation}\label{Processing delay}
  \ T_i = \begin{cases}
  \max \left \lbrace T_i^{local}, T_{i,P}^{ES}, T_{i,P}^{CS}  \right \rbrace ,  \quad\quad\quad\quad\ \  \text{if } l_i = 1;  \vspace{1ex} \\
  \sum_{j=1}^M S_{i,j} \max \left \lbrace T_i^{local}, T_{i,j}^{ES}, T_{i,j}^{CS} \right\rbrace , \quad \text{otherwise},
\end{cases}
\end{equation}
where $T_{i,P}^{ES} = T_{i,P}^{comm} + T_{i,P}^{comp}$,\vspace{1ex} $T_{i,P}^{CS} = T_{i,P}^{comm} + T_{i,CS}^{comm} + T_{i,CS}^{comp}$, $T_{i,j}^{ES} = T_{i,j}^{move} + T_{i,j}^{comm} + T_{i,j}^{comp}$ and\vspace{1ex} $T_{i,j}^{CS} = T_{i,j}^{move} + T_{i,j}^{comm} + T_{i,CS}^{comm} + T_{i,CS}^{comp}$. $T_{i,j}^{move}$\vspace{1ex} is the time vehicle $V_i$ moves to the coverage of RSU $E_j$, which can be formulated as:
\begin{equation}\label{Travel Time}
  \ \displaystyle{T_{i,j}^{move} = \left [\frac {x_j - x_i - r_j}{v_i}\right]^+} \, ,
\end{equation}
where $x_j$, $x_i$ represent the coordinates of RSU $E_j$ and vehicle $V_i$ respectively, and $[\cdot]^+$ denotes $\max \lbrace 0,\cdot \rbrace$.

\subsection{System Utility}

The objective of this paper is to maximize the system utility of the platform. The platform rents servers to provide computation resource and utilizes the resource to process the tasks from vehicles. Apart from profits, QoS is also one of the important indicators for the system utility. Therefore, we design the system utility that considers both the profits and QoS.

For the Cloud-Edge-Terminal collaborative vehicular networks, the main factor that affects the QoS is the processing delay. Generally, QoS should monotonically decrease with the increase of the processing delay. According to~\cite{Dai:2019}, we choose the logarithmic function as the QoS model. Thus, the QoS utility for vehicle $V_i$, denoted as $U_i^{QoS}$, is:
\begin{equation}\label{QoS}
\ U_i^{QoS} = \beta_{Q} \log\Big(1 + \varepsilon - T_i\Big) \, ,
\end{equation}
where $\beta_{Q}$ is the weight of QoS, and $\varepsilon$ is a known value that guarantees $U_i^{QoS} > 0$. Thus, the system utility for vehicle $V_i$, denoted as $ U_i$, is formulated as:
\begin{equation}\label{system utility function}
\begin{split}
\ U_i &= \beta_P \Big(c_v {(\alpha_{i,e} + \alpha_{i,c})}D_i R_i - c_s{(f_i + f_{i,CS})}\Big) + U_i^{QoS} \\
&= \beta_P \Big(c_v {(\alpha_{i,e} + \alpha_{i,c})}D_i R_i - c_s{(f_i + f_{i,CS})}\Big) \\
&\quad + \beta_{Q} \log\Big(1 + \varepsilon - T_i\Big) \, ,
\end{split}
\end{equation}
where $\beta_P$ is the weight of profits, $c_v$ and $c_s$ represent the unit price that system charges from vehicles for the task processing and the unit price that the system rents the servers respectively, $f_i$ (in CPU cycles/s) denotes the allocated computation resource to vehicle $V_i$ from ES or virtual resource pool, i.e.,

\begin{equation}
\ f_i = \begin{cases}
f_{i,P} ,  \ \    \quad\quad\quad\ \text{if } l_i = 1; \vspace{1ex} \\
f_{i,j}, \ \ \ \ \ \ \ \ \ \ \ \ \text{otherwise}.
\end{cases}
\end{equation}
\subsection{Problem Formulation}

\begin{table}[]
\caption{List of notations}
\label{List of notations}
\begin{spacing}{0.63}
\begin{tabular}{lm{6.7cm}}
\hline
\textbf{Symbols}             & \textbf{Definitions}                                                    \\ \hline
\multicolumn{2}{l}{\textbf{Sets:}}                                                                     \\
$\mathcal M$                 & Set of RSUs.                                                            \\
$\mathcal N$                 & Set of vehicles.                                                        \\
$r$                          & Set of RSUs' radius.                                                    \\
\multicolumn{2}{l}{\textbf{Decision Variables:}}                                                       \\
$\alpha_i$                   & Offloading ratio of vehicle $V_i$.                                      \\
$S_{i,j}$                    & Offloading decision of vehicle $V_i$ on RSU $E_j$.                      \\
{$f_i$}                      & Allocated computation resource to vehicle $V_i$                         \\
                             & from ES or virtual resource pool.                                       \\
\multicolumn{2}{l}{\textbf{Functions:}}                                                                \\
$T_{i,j}^{move}$             & The time vehicle $V_i$ moves to the coverage of                         \\
                             & RSU $E_j$.                                                              \\
$T_i^{local}$                & Local computing delay of vehicle $V_i$.                                 \\
$T_{i,CS}^{comm}$            & Communication delay for cloud computing                                 \\
                             & between vehicle $V_i$ and CS.                                           \\
$T_{i,j}^{comm}$             & Communication delay for edge computing                                  \\
                             & between vehicle $V_i$ and RSU $E_j$.                                    \\
$T_{i,P}^{comm}$             & Communication delay for edge computing be-                              \\
                             & tween vehicle $V_i$ and the virtual resource pool.                      \\
$T_{i,CS}^{comp}$            & Cloud computation delay of vehicle $V_i$.                               \\
$T_{i,j}^{comp}$             & Computation delay for edge computing in                                 \\
                             & general region when vehicle $V_i$ offloads task                         \\
                             & to RSU $E_j$.                                                           \\
$T_{i,P}^{comp}$             & Computation delay for edge computing in                                 \\
                             & overlapping region of vehicle $V_i$.                                    \\
$T_i$                        & Processing delay of vehicle $V_i$.                                      \\
$U_i^{QoS}$                  & QoS utility for vehicle $V_i$.                                          \\
$U_i$                        & System utility for vehicle $V_i$.                                       \\
\multicolumn{2}{l}{\textbf{Parameters:}}                                                               \\
$l_i$                        & The region that vehicle $V_i$ is located in.                            \\
$D_i$                        & The size of task input data of vehicle $V_i$.                           \\
$R_i$                        & The number of CPU cycles required to process                            \\
                             & 1-bit of the task of vehicle $V_i$.                                     \\
$T_i^{max}$                  & The maximum tolerant delay to accomplish                                \\
                             & the task of vehicle $V_i$.                                              \\
$\rho_{i,CS}$                & Data transmission rate between vehicle $V_i$                            \\
                             & and CS.                                                                 \\
$\rho_{i,j}$                 & Data transmission rate between vehicle $V_i$                            \\
                             & and RSU $E_j$.                                                          \\
$\rho_{i,P}$                 & Data transmission rate between vehicle $V_i$                            \\
                             & and the virtual resource pool.                                          \\
$f_{i,V}$                    & Local computation resource of vehicle $V_i$.                            \\
{$f_{i,CS}$}                 & Allocated computation resource to vehicle $V_i$                         \\
                             & from CS.                                                                \\
$F_j$                        & Total computation resource of ES in RSU $E_j$.                          \\
$F_{P}$                      & Total computation resource of resource pool.                            \\
$\beta_P$                    & Weight of profits.                                                      \\
$\beta_Q$                    & Weight of QoS.                                                          \\
$c_v$                        & Unit price that system charges from vehicles                            \\
                             & for the task processing.                                                \\
$c_s$                        & Unit price that system rents the servers.                               \\
$\varepsilon$                & A known value that guarantees $U_i^{QoS} > 0$.                          \\ \hline
\end{tabular}
\end{spacing}
\end{table}

We formulate the offloading schedule as an optimization problem and aim to maximize the system utility. Let \bm{$S$} $= \lbrace S_{i,j}\rbrace$ denote the selection decision on RSUs, \bm{$\alpha$} $= \lbrace \bm{\alpha_c}, \bm{\alpha_e} \rbrace$ denote the offloading ratio decision of the vehicles, \bm{$f$} $= \lbrace f_i\rbrace$ denote the resource allocation decision for the vehicles. Then, we formulate the offloading schedule problem as follows:
\begin{align}
P1: \underset{\bm{S},\bm{\alpha},\bm{f}}{max} &\ \sum_{i=1}^N U_i \, , \label{Problem 1} \\
 s.t.\quad  &T_i \leq T_i^{max}\quad \forall i \in [1,\cdots,N]; \tag{\ref{Problem 1}{{a}}}\label{Problem 1.1}\\
 & 0 \leq \alpha_{i,e} + \alpha_{i,c} \leq 1\quad \forall i \in [1,\cdots,N]; \tag{\ref{Problem 1}{{b}}}\label{Problem 1.2}\\
 & 0 \leq \alpha_{i,e} \leq 1\quad \forall i \in [1,\cdots,N]; \tag{\ref{Problem 1}{{c}}}\label{Problem 1.3}\\
 & 0 \leq \alpha_{i,c} \leq 1\quad \forall i \in [1,\cdots,N]; \tag{\ref{Problem 1}{{d}}}\label{Problem 1.4}\\
 & S_{i,j} = \lbrace 0,1\rbrace\quad \forall i \in [1,\cdots,N] , \ j \in [1,\cdots,M_G]; \tag{\ref{Problem 1}{{e}}}\label{Problem 1.5}\\
 & \sum_{j=1}^M S_{i,j} = 1\quad \forall i \in [1,\cdots,N]; \tag{\ref{Problem 1}{{f}}}\label{Problem 1.6}\\
 & 0 \leq \sum_{i=1}^N S_{i,j} f_{i,j} \leq F_j \quad \forall j \in [1,\cdots,M_G]; \tag{\ref{Problem 1}{{g}}}\label{Problem 1.7}\\
 & 0 \leq \sum_{i=1}^N f_{i,P} \leq F_{P}, \tag{\ref{Problem 1}{{h}}}\label{Problem 1.8}
\end{align}
Constraint (\ref{Problem 1.1}) ensures that the task of vehicle $V_i$ can be accomplished within the maximum tolerant delay $T^{max}_i$. Then Constraints (\ref{Problem 1.2}), (\ref{Problem 1.3}) and (\ref{Problem 1.4}) guarantee that vehicle $V_i$ can only offload a ratio of its task. Constraints (\ref{Problem 1.5}) and (\ref{Problem 1.6}) show that vehicle $V_i$ can select only one RSU as its offloading target. Constraints (\ref{Problem 1.7}) and (\ref{Problem 1.8}) ensure that the allocated computation resource to the vehicles is no more than the total resource. Moreover, Constraint (\ref{Problem 1.5}) show that the offloading decisions (i.e., \bm{$S$}) are binary, while the objective function of Problem P1 is non-linear with respect to \bm{$S$}. Thus, Problem P1 is a mixed integer non-linear programming problem, which is NP-hard~\cite{Garey:1990}.

For ease of reference, we list the frequently used notations in TABLE~\ref{List of notations}.

\section{The Proposed MSCET Schedule}\label{sec4}

In this section, we introduce the proposed multi-scenario offloading schedule for biomedical data processing and analysis in the Cloud-Edge-Terminal collaborative vehicular networks (MSCET). As mentioned above, there are two different cases for the proposed MSCET schedule (i.e., overlapping region and general region), we will discuss each of them in details.

\subsection{Decision in Overlapping Region}
\vspace{-1mm}
In overlapping region, the virtual resource pool is constructed to integrate computation resource, and all vehicles in the region offload tasks to the resource pool. Thus, {\bf$S$} requires no consideration. However, the remaining decision vectors (i.e., {\bf$\alpha_{c}$}, {\bf$\alpha_{e}$} and {\bf$f$}) of the optimization problem in Eq.~(\ref{Problem 1}) are highly coupled with each other, which makes it more difficult to solve the problem. Decoupling the vectors is one of the feasible solutions. Specially, we can fix two of them as constants and then calculate the third one. The two vectors we fixed can be updated based on the result of calculation. The process can be repeated until convergence.

\subsubsection{Optimization of Offloading Ratio for CS}

As {\bf$\alpha_{c}$} is relatively weakly correlated with other vectors, its value will be firstly determined. Due to the assumption that {\bf$\alpha_{e}$} and {\bf$f$} are constants, the delay of edge computing is also a constant. Thus, $T_i$ can be reformulated as:
\begin{equation}\label{Processing delay for CS}
\ T_i = \max \left \lbrace T_i^{local}, T_{i,P}^{comm} + T_{i,CS}^{comm} + T_{i,CS}^{comp}  \right \rbrace .
\end{equation}

Since $T_i$ is affected by but non-differentiable with respect to {\bf$\alpha_{c}$}, we first approximate $T_i$ as:
\begin{equation}\label{Approximate processing delay for CS}
\begin{split}
\ T_i &\leq T_i^{local} + T_{i,P}^{comm} + T_{i,CS}^{comm} + T_{i,CS}^{comp} \\
& = \frac {\Big(1 - \alpha_{i,e} - \alpha_{i,c}\Big) D_iR_i}{f_{i,V}} + \frac{\alpha_{i,e}D_i}{\rho_{i,P}} + \frac{\alpha_{i,c}D_i}{\rho_{i,CS}} + \frac {\alpha_{i,c}D_iR_i}{f_{i,CS}} \\
& = \alpha_{i,c} \omega_{i,c} +  \delta_{i,c}  \, ,
\end{split}
\end{equation}
where $\displaystyle{\omega_{i,c} = \frac{D_i}{\rho_{i,CS}} + \frac {D_iR_i}{f_{i,CS}} - \frac {D_iR_i}{f_{i,V}}}$, \vspace{1ex} $\displaystyle{\delta_{i,c} = \frac{\alpha_{i,e}D_i}{\rho_{i,P}} + \frac {(1 - \alpha_{i,e})D_iR_i}{f_{i,V}}}$.

Let $T_i$ be the value of its upper bound, which represents the worst case. By substituting Eq.~(\ref{Approximate processing delay for CS}) into Problem P1, it can be reformulated as:
\begin{align}
P2.1: \underset{{\bf\alpha_{c}}}{max} & \sum_{i=1}^N \beta_P  \Big(c_v (\alpha_{i,c} + \alpha_{i,e}) D_i R_i - c_s (f_{i,P} + f_{i,CS})\Big)  \notag \\
 &\quad\ + \beta_Q \log\Big(1 + \varepsilon  - \alpha_{i,c} \omega_{i,c} -  \delta_{i,c}\Big) \, , \label{Problem 2} \\
 s.t.\quad  & \alpha_{i,c} \omega_{i,c} +  \delta_{i,c} \leq T_i^{max}\quad \forall i \in [1,\cdots,N]; \tag{\ref{Problem 2}{{a}}}\label{Problem 2.1}\\
 & 0 \leq \alpha_{i,c} \leq 1 - \alpha_{i,e} \quad \forall i \in [1,\cdots,N]. \tag{\ref{Problem 2}{{b}}}\label{Problem 2.2}
\end{align}

\newtheorem{lemma}{Lemma}
\begin{lemma}\label{lm1}
Problem P2.1 is a convex optimization problem.
\end{lemma}

\begin{proof}
The second-order derivative of $U(\bm{\alpha_c})$ with respect to $\bm{\alpha_c}$ is:
\begin{equation}\label{Lemma 1}
\frac{\partial^2 U(\bm{\alpha_c})}{\partial \bm{\alpha_c}^2} = - \frac{\beta_Q{\omega_{i,c}}^2}{\Big(1 + \varepsilon - \alpha_{i,c} \omega_{i,c} - \delta_{i,c}\Big)^2\ln2} \leq 0,
\end{equation}
Thus the objective function $U(\bm{\alpha_c})$ is convex. Combining with the linear convex Constraints (\ref{Problem 2.1}) and (\ref{Problem 2.2}), Problem P2.1 is a convex optimization problem according to the definition of convex optimization problem~\cite{Boyd:2004}.
\end{proof}

According to Lemma~\ref{lm1}, Problem P2.1 is a convex optimization problem. Here we use genetic algorithm to solve it.

\subsubsection{Optimization of Computation Resource}

After determination of \bm{$\alpha_c$}, we need to update the values of \bm{$\alpha_e$} and \bm{$f$}. Based on the above decoupling method, the formulation of computation resource optimization under given \bm{$\alpha_c$} and \bm{$\alpha_e$} is:
\begin{align}
 P2.2:\underset{\bm{f}}{max} & \sum_{i=1}^N \beta_P  \Big(c_v (\alpha_{i,c} + \alpha_{i,e}) D_i R_i - c_s (f_{i,P} + f_{i,CS})\Big)  \notag \\
 &\quad\ + \beta_Q \log\Big(1 + \varepsilon  - T_i\Big)  \, , \label{Problem 3} \\
 s.t. \quad &\sum_{i = 1}^N f_{i,P} \leq F_{P}; \tag{\ref{Problem 3}{{a}}}\label{Problem 3.1}\\
  & \frac{\alpha_{i,e}D_i}{\rho_{i,P}} + \frac{\alpha_{i,e}D_iR_i}{f_{i,P}} \leq T_i^{max}\quad \forall i \in [1,\cdots,N]. \tag{\ref{Problem 3}{{b}}}\label{Problem 3.2}
\end{align}
Note that the local computing delay and cloud computing delay are all constants under given \bm{$\alpha_c$} and \bm{$\alpha_e$}, we update the Constraint (\ref{Problem 1.1}) as Constraint (\ref{Problem 3.2}).

The objective function $U(\bm{f})$ in Problem P2.2 is convex since $\displaystyle{\frac{\partial^2 U(\bm{f})}{\partial \bm{f}^2} = - \frac{ 2 \beta_Q \alpha_{i,e} D_i R_i }{f_{i,P}^3 \Big(1 + \varepsilon - \frac{\alpha_{i,e}D_i}{\rho_{i,P}} - \frac{\alpha_{i,e}D_iR_i}{f_{i,P}} \Big)\ln2}}$ $\displaystyle{- \frac {\beta_Q (\alpha_{i,e} D_i R_i)^2}{f_{i,P}^4 \Big(1 + \varepsilon - \frac{\alpha_{i,e}D_i}{\rho_{i,P}} - \frac{\alpha_{i,e}D_iR_i}{f_{i,P}} \Big)^2\ln2}  \leq 0}$\vspace{1ex}, and Constraints (\ref{Problem 3.1}) and (\ref{Problem 3.2}) are linear convex. According to Lemma~\ref{lm1}, Problem P2.2 is a convex optimization problem and can be solved by interior point method. The first step of interior point method is to construct a penalty function, which can be expressed as:
\begin{equation}\label{interior point method}
\begin{split}
\ \varphi\Big(\bm{f},r^{(k)}\Big) =& \ U\Big(\bm{f}\Big) \\
& - r_1^{(k)} \sum_{i=1}^N \frac{1}{{\frac{\alpha_{i,e}D_i}{\rho_{i,P}} + \frac{\alpha_{i,e}D_iR_i}{f_{i,P}} - T_i^{max}}}\\
& - r_2^{(k)} \frac{1}{\sum_{i = 1}^N f_{i,P} - F_{P}} \, ,
\end{split}
\end{equation}
where $r^{(k)}$ are penalty factors and obey the law of decline, i.e., $r^{(0)} > r^{(1)} > r^{(2)} > \cdots > r^{(k)} > r^{(k+1)} > \cdots > 0$ and $\underset{k\to\infty}{\lim}r^{(k)} = 0$.

The second step is setting $\displaystyle{\frac{\partial \varphi}{\partial \bm{f}} = 0}$ \vspace{1ex} to find the extreme point $\bm{f^*}$ of the penalty function. Then substitute $\bm{f^*}$ into the condition of convergence $ \Vert \,f^*(r^k) - f^*(r^{k-1})\, \Vert \leq \varepsilon_1$, where $\varepsilon_1$ is the threshold that we set in advance. If $\bm{f^*}$ satisfies the above condition, the process will be finished, and we can obtain $\bm{f} = \bm{f^*}$. Otherwise, let $k = k + 1$ and update $r^{(k)}$ continuously until satisfying the above condition.

\subsubsection{Optimization of Offloading Ratio for ES}

Similar to the optimization of offloading ratio for CS, $T_i$ can be approximated as:
\begin{equation}\label{Approximate processing delay for ES}
\begin{split}
\ T_i &\leq \frac {T_i^{local} + T_{i,P}^{comm} + T_{i,P}^{comp} + T_{i,P}^{comm} + T_{i,CS}^{comm} + T_{i,CS}^{comp}}{\lambda} \\
& = \frac {\Big(1 - \alpha_{i,e} - \alpha_{i,c}\Big) D_iR_i}{\lambda f_{i,V}} + \frac{\alpha_{i,e}D_i}{\lambda\rho_{i,P}} + \frac{\alpha_{i,e}D_iR_i}{\lambda f_{i,P}} \\
& \quad + \frac{\alpha_{i,e}D_i}{\lambda\rho_{i,P}} + \frac{\alpha_{i,c}D_i}{\lambda\rho_{i,CS}} + \frac {\alpha_{i,c}D_iR_i}{\lambda f_{i,CS}} \\
& = \alpha_{i,e} \omega_{i,e} +  \delta_{i,e}  \, ,
\end{split}
\end{equation}
where $\displaystyle{\omega_{i,e} = \frac{2D_i}{\lambda \rho_{i,P}} + \frac {D_iR_i}{\lambda f_{i,P}} - \frac {D_iR_i}{\lambda f_{i,V}}}$, $\displaystyle{\delta_{i,e} = \frac{\alpha_{i,c}D_i}{\lambda \rho_{i,CS}} + \frac {(1 - \alpha_{i,c})D_iR_i}{\lambda f_{i,V}}} + \frac {\alpha_{i,c}D_iR_i}{\lambda f_{i,CS}}$, and $\lambda$ is a known value to adjust the approximation degree. Then we can reformulate Problem P1 as:
\begin{align}
 P2.3:\underset{\bm{\alpha_e}}{max}\ &\sum_{i=1}^N \ \beta_P  \Big(c_v (\alpha_{i,c} + \alpha_{i,e}) D_i R_i - c_s (f_{i,P} + f_{i,CS})\Big) \notag  \\
 & \quad\ \; + \beta_{Q} \log\Bigg(1 + \varepsilon - \alpha_{i,e} \omega_{i,e} -  \delta_{i,e}\Bigg) \, , \label{Problem 4} \\
 s.t.\quad  & 0 \leq \alpha_e \leq 1 - \alpha_c \quad \forall i \in [1,\cdots,N]; \tag{\ref{Problem 4}{{a}}}\label{Problem 4.1} \\
 \quad &\alpha_{i,e} \omega_{i,e} +  \delta_{i,e} \leq T_i^{max}\quad \forall i \in [1,\cdots,N]. \tag{\ref{Problem 4}{{b}}}\label{Problem 4.2}
\end{align}

The objective function $U(\bm{\alpha_e})$ of Problem P2.3 is convex since $\displaystyle{\frac{\partial^2 U(\bm{\alpha_e})}{\partial \bm{\alpha_e}^2} = - \frac{\beta_Q{\omega_{i,e}}^2}{\Big(1 + \varepsilon - \alpha_{i,e} \omega_{i,e} - \delta_{i,e}\Big)^2\ln2} \leq 0}$\vspace{1ex}, and Constraints (\ref{Problem 4.1}), (\ref{Problem 4.2}) are linear convex, thus the Problem P2.3 is a convex optimization problem according to Lemma~\ref{lm1}. The Lagrangian method~\cite{Sorkhoh:2019} is adopted to solve this problem. The Lagrangian function can be expressed as:
\begin{equation}\label{Lagrangian function}
\begin{split}
 \mathcal{L}_i(\bm{\alpha_e},\zeta,\eta,\theta) = & \ \beta_P  \Big(c_v (\alpha_{i,c} + \alpha_{i,e}) D_i R_i - c_s (f_{i,P} + f_{i,CS})\Big) \\
 & + \beta_{Q} \log\Bigg(1 + \varepsilon - \alpha_{i,e} \omega_{i,e} -  \delta_{i,e}\Bigg) \\
 & + \zeta_i\Big(\alpha_{i,e} \omega_{i,e} +  \delta_{i,e} - T_i^{max}\Big) \\
 & + \eta_i\Big(\alpha_{i,e} + \alpha_{i,c} - 1\Big) + \theta_i\Big(0 - \alpha_{i,e}\Big) \, ,
\end{split}
\end{equation}
where $\zeta$, $\eta$ and $\theta$ are Lagrangian multipliers. According to KKT condition~\cite{Boyd:2004}, the following conditions must be satisfied:
\begin{equation}\label{KKT condition}
\left\{\begin{lgathered}
 \zeta_i\Bigg(\alpha_{i,e} \omega_{i,e} +  \delta_{i,e} - T_i^{max}\Bigg) = 0; \\
 \eta_i\Big(\alpha_{i,e} + \alpha_{i,c} - 1\Big) = 0; \\
 \theta_i\Big(0 - \alpha_{i,e}\Big) = 0; \\
 \zeta_i > 0,\ \eta_i \geq 0,\ \theta_i \geq 0; \\
 \frac{\partial \mathcal{L}_i}{\partial \alpha_{i,e}} = 0.
\end{lgathered} \right.
\end{equation}

As mentioned in Constraint (\ref{Problem 4.1}), we know that $0 \leq \alpha_{i,e} \leq 1 - \alpha_{i,c} $. Thus, it is simple to solve $\eta_i(\alpha_{i,e} + \alpha_{i,c} - 1) = 0$ and $\theta_i(0 - \alpha_{i,e}) = 0$. The related discussion is as follows:

Case 1: $\eta_i > 0$ and $\theta_i > 0$. We get $\alpha_{i,e} = 0 $ and $\alpha_{i,c} = 1 $, i.e., vehicle $V_i$ offloads the whole task to the CS.

Case 2: $\eta_i > 0$ and $\theta_i = 0$. We get $\alpha_{i,e} = 1 - \alpha_{i,c} $, i.e., vehicle $V_i$ offloads the whole task to its selected ES and CS.

Case 3: $\eta_i = 0$ and $\theta_i > 0$. We get $\alpha_{i,e} = 0$, i.e., vehicle $V_i$ offloads $\alpha_{i,c}D_i$ part of its task to the CS and processes the rest locally.

Case 4: $\eta_i = 0$ and $\theta_i = 0$. We cannot get the unique solution of $\alpha_{i,e}$. Then it is necessary to utilize other conditions, since $\zeta_i > 0$, we can get:
\begin{align}\label{KKT condition 3}
& \alpha_{i,e} \omega_{i,e} +  \delta_{i,e} - T_i^{max} = 0 \notag \\
\Rightarrow \quad & \alpha_{i,e} = \frac{ T_i^{max} - \delta_{i,e}}{\omega_{i,e}} \, .
\end{align}

Then $\zeta_i$ can be calculated as:
\begin{align}\label{KKT condition 4}
& \frac{\partial \mathcal{L}_i}{\partial \alpha_{i,e}} = \beta_P c_v D_i R_i - \frac{\beta_Q \omega_{i,e}}{\Big(1 + \varepsilon - \alpha_{i,e} \omega_{i,e} -  \delta_{i,e}\Big)\ln2} + \zeta_i\omega_{i,e} = 0 \notag \\
\Rightarrow \quad & \zeta_i =  \frac{\beta_Q}{\Big(1 + \varepsilon - T_i^{max}\Big)\ln2} - \frac{\beta_P c_v D_iR_i}{\omega_{i,e}}  \,
\end{align}

And we can get $\alpha_i$ as follows:
\begin{equation}\label{KKT condition 5}
\ \alpha_{i,e} = \begin{cases}
\displaystyle{\frac{ T_i^{max} - \delta_{i,e}}{\omega_{i,e}}}, & \text{if } \eta_i = 0 \text{ and } \theta_i = 0; \\
1 - \alpha_{i,c}, & \text{if } \eta_i > 0 \text{ and } \theta_i = 0; \\
0, &\text{otherwise}.
\end{cases}
\end{equation}

\subsection{Decision in general region}

Different from overlapping region, each vehicle needs to select one RSU as its offloading target in general region. Therefore, it is necessary to calculate \bm{$S$}. The matching method is adopted to get the optimal selection. Determining the appropriate matching weight between vehicles and RSUs is one of the challenges. The processing delay for edge computing in general region includes the moving time, which is affected by the distance between vehicles and RSUs. Thus, we use the distance as the weight of matching method.

The first step is to construct a weighted bipartite graph $\mathcal G(\mathcal U, \mathcal V)$ to establish the relation between vehicles and RSUs. The set $\mathcal U$ represents the vehicles in general region, and the set $\mathcal V$ represents the RSUs along the unidirectional road.

The second step is to utilize Kuhn-Munkres algorithm to get the minimum weight matching after which \bm{$S$} can be obtained. Specifically, $S_{i,j} = 1$ if vehicle $V_i$ selects RSU $E_j$ as its offloading target, and $S_{i,j} = 0$ otherwise.

According to the selection decision \bm{$S$}, there may be some vehicles that need to move to the communication coverage of their selected RSUs. Cloud-Terminal cooperation is utilized for the movement to reduce the total task processing delay. Due to the fact that ESs are resource-limited, vehicle $V_i$ will offload as high ratio of its task as possible to CS to save the resource of ESs. Thus, we set $T_{i,CS}^{comm^*} = T^{move}_{i,j}$, which represents the maximum offloading ratio for CS. Then the task characteristics of vehicle $V_i$ can be updated as:
\begin{align}
&D_i^* = D_i - \frac{T^{move}_{i,j}f_{i,V}}{R_i} - T^{move}_{i,j}\rho_{i,CS} ; \label{D}   \\
&T_i^{max^*} = T_i^{max} - T^{move}_{i,j}. \label{T}
\end{align}

The decisions of \bm{$\alpha_e$}, \bm{$\alpha_c$} and \bm{$f$} of the new optimization problem are similar to the decisions in overlapping region, so we do not repeat them again.

\subsection{The MSCET Schedule Algorithm}

The pseudo of MSCET schedule algorithm is shown in Algorithm~\ref{alg1}. Vehicle $V_i$ will send the message about the characteristics $\lbrace D_i,R_i,T_i^{max} \rbrace$ of its task to the system when it cannot process the task in time. And then the system calculates the offloading schedule based on the messages from all the vehicles in the area. Specifically, there are two cases according to the location of each vehicle. In general region, each vehicle needs to select one RSU as its offloading target by Kuhn-Munkres matching method while in overlapping region the RSU selection is unnecessary. Then, Cloud-Terminal cooperation is utilized to improve the system utility during the moving time in general region. After that, Cloud-Edge-Terminal cooperation is introduced to further improve the utility. Last but not the least, the genetic algorithm, interior point method and KKT condition are adopted to optimize \bm{$\alpha_c$}, \bm{$f$} and \bm{$\alpha_e$} until getting the final result.

The complexities of initialization, Phase 1 and Phase 2 are $\mathcal{O}( 3N + 3 )$, $\mathcal{O}( N^4 )$ and $\mathcal{O}( 3N )$ respectively. While for $k$ iterations, the complexity of the inner loop is $\mathcal{O}( k( 2N + 1) )$. Then the total complexity of middle loop, for $z$ iterations, is $\mathcal{O}( z (k( 2N + 1) + N))$. After that, for $s$ iterations, the complexity of Phase 3 is $\mathcal{O}(sN^2( z (k( 2N + 1) + N)))$. Considering the dominant term, the time complexity of Algorithm~\ref{alg1} is $\mathcal{O}( N^4 + 2skzN^3)$.

\begin{algorithm}[!h]
\caption{MSCET Schedule Algorithm}
\begin{algorithmic}[1]\label{alg1}
\REQUIRE
$D_i$: the size of task input data of vehicle $V_i$; \\
\quad\;\; $R_i$: {the number of CPU cycles required to process\\
\quad\quad\quad\ 1-bit of the task of vehicle $V_i$;} \\
\quad\;\; $T^{max}_i$: the maximum tolerant delay of vehicle $V_i$'s task.
\ENSURE
Offloading schedule consisting of \bm{$S$}, \bm{$\alpha_c$}, \bm{$\alpha_e$}and \bm{$f$}.
\STATE Initialization: \\
\quad Set initial feasible solution $\alpha_e^{(0)}$, $\alpha_c^{(0)}$, $f^{(0)}$; \\
\quad Set iteration number $s = 1$, $z = 1$, $k = 1$;
\IF {$l_i = 0$}
\STATE /*Phase 1: Selection of RSU*/
\STATE Get the minimum weight matching and $S_{i,j}$;
\STATE /*Phase 2: Cloud-Terminal Cooperation*/ \vspace{1ex}
\STATE Let $\displaystyle T_{i,CS}^{comm^*} = T^{move}_{i,j}$ \vspace{1ex}and update $D_i$ and $T^{max}_i$ based on Eq.~(\ref{D}) and Eq.~(\ref{T});
\ENDIF
\STATE /*Phase 3: Cloud-Edge-Terminal Cooperation*/
\STATE /*Phase 3.1: Optimization of offloading ratio for CS*/
\REPEAT
\STATE Utilize genetic algorithm to optimize $\alpha_{i,c}$;
\STATE $s = s + 1$;
\STATE /*Phase 3.2: Optimization of offloading ratio for ES*/
\REPEAT
\STATE /*Phase 3.3: Optimization of computation resource*/
\REPEAT
\STATE Construct the penalty function $\varphi\Big(f,r^{(k)}\Big)$ based on Eq.~(\ref{interior point method}) and set $\displaystyle{\frac{\partial \varphi}{\partial f} = 0}$ \vspace{1ex} to find the extreme point $\bm{f^*}$ \vspace{1ex} of the penalty function;
\STATE Substitute $\bm{f^*}$ \vspace{1ex} into the convergence condition \\
$\displaystyle \Vert \,f^*(r^k) - f^*(r^{k-1})\, \Vert \leq \varepsilon_1$ \vspace{1ex} to verify its feasibility;
\STATE $k = k + 1$ and update $\displaystyle r^{(k)}$;
\UNTIL Satisfaction
\STATE Calculate $\alpha_{i,e}^{(z)}$ based on Eq.~(\ref{KKT condition 5});
\STATE $z = z + 1$;
\UNTIL Convergence
\UNTIL Convergence
\end{algorithmic}
\end{algorithm}

\section{Performance Evaluation}\label{sec5}

In this section, we evaluate the proposed multi-scenario offloading schedule (MSCET) and compare it with two benchmark schedules and two different modes. We first describe the simulation settings, and then show the schedules' reliability. Afterwards we conduct the performance comparisons to illustrate the effectiveness of the proposed MSCET.

\subsection{Simulation Settings}

As mentioned, there are two scenarios in the vehicular network (i.e., overlapping region and general region). In general region, we consider a 250-meter unidirectional road, along which 5 RSUs are deployed. Each RSU is equipped with an ES, which has the computation resource of 0.5 GHz. While in overlapping region, the virtual resource pool is constructed to integrate resource of ESs, and the total resource of each virtual resource pool is 2.5 GHz. There are some vehicles running at a speed of 40 km/h. Each vehicle $V_i$ has a computation-intensive or delay-sensitive task with three characteristics $\lbrace D_i,R_i,T_i^{max} \rbrace$. We set $D_i$ and $T_i^{max}$ to be randomly distributed in the range of $[10,15]$MB and $[4,6]$s respectively. The local computation resource of each vehicle is 12.5 MHz.

For performance comparison, we consider the following two benchmarks schedules for Cloud-Edge-Terminal cooperation and two cooperation modes using our proposed MSCET schedule.

\textit{1) Given Ratio and Resource based Selection optimization (SGRR)}: The schedule selects the optimal offloading target for each vehicle, under a given offloading ratio and computation resource.

\textit{2) Nearby}: The schedule optimizes computation resource and offloading ratio, after which the vehicles offload their tasks to nearby RSUs.

\textit{3) Cloud-Terminal}: Each task is jointly processed in the terminal and cloud server.

\textit{4) Edge-Terminal}: Each task is jointly processed in the terminal and edge server.

\subsection{Convergence}

To verify the reliability of the proposed MSCET schedule, we evaluate the convergence of algorithm under different initial points. The algorithm can converge only if the inner loop can converge, thus the convergence of inner loop is omitted here and only the convergence of algorithm is plotted, i.e. the convergence of system utility. Fig.~\ref{iteration} shows that the system utility can converge within 6 iterations and almost to the same value under different initial points, indicating the reliability of the proposed MSCET schedule.

\begin{figure}
  \centering
  \includegraphics[width=2.9in]{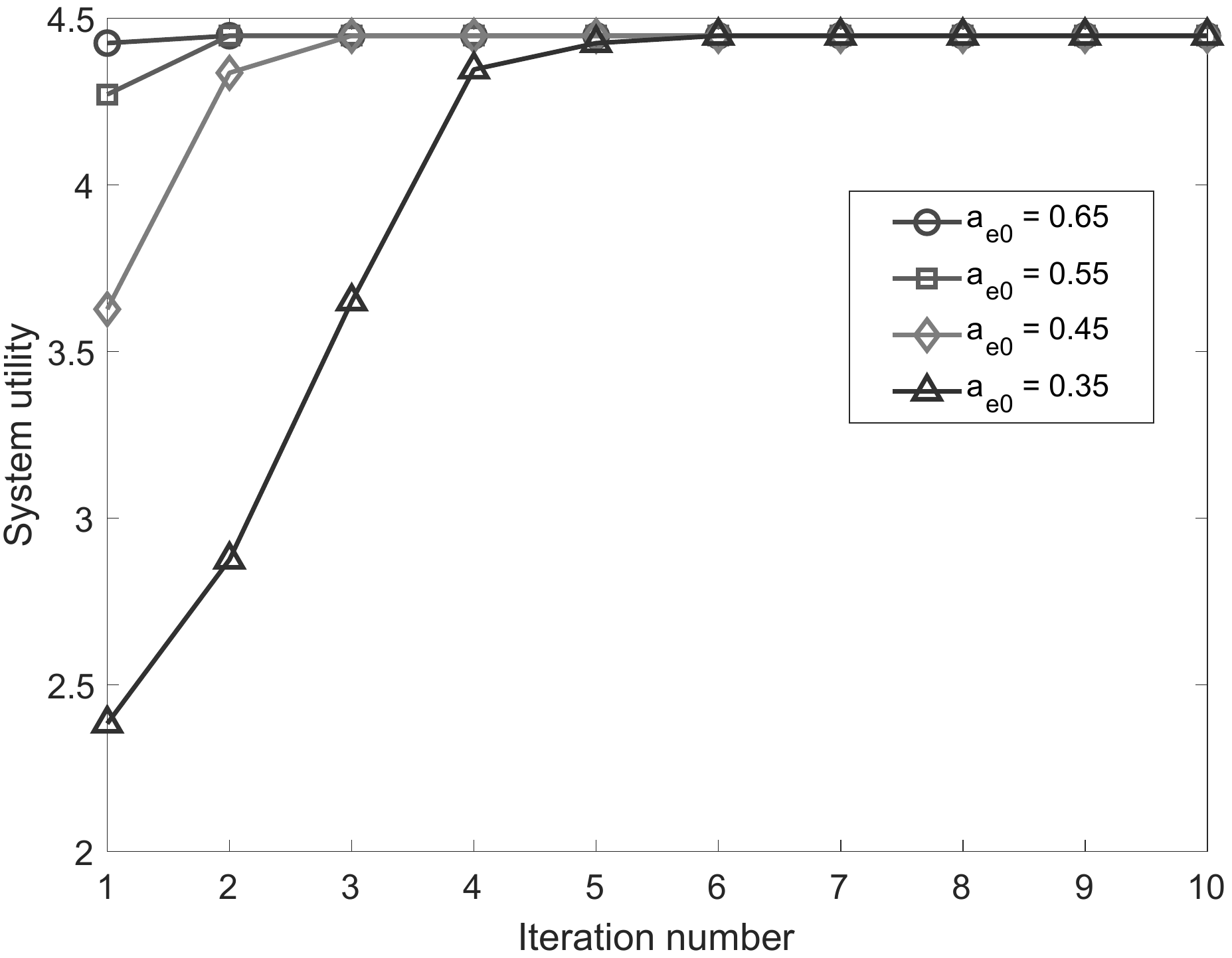}
  \caption{Convergence of system utility under different initial points.}\label{iteration}
\end{figure}

\subsection{Impact of RSUs' Coverage Size}

To compare the scalability of different schedules, we study the impact of the radius of RSUs' coverage on system utility in Fig.~\ref{RSUsize}. The results illustrate that the system utility of the Nearby will be decreased with the increasing radius of RSUs' coverage, while the coverage size has no impact on the system utility of the other two schedules, which means that the two schedules have a better scalability. The reason is that both the MSCET and SGRR consider the offloading selection optimization, making them not sensitive to the increase of coverage size. However, for the Nearby, with the expansion of RSUs' coverage, the increasing number of vehicles in RSUs' coverage may cause that more vehicles connect to the same RSU, which may lead to the overload of RSU. In this case some tasks cannot be processed in the current time slot, and the system utility may become lower. Besides, the performance of the Nearby may fluctuate when some vehicles are located in the boundary of RSUs, shown as the fluctuation in Fig.~\ref{RSUsize}. Moreover, the MSCET outperforms the SGRR due to its joint optimization of resource allocation and offloading selection.

\begin{figure}
  \centering
  \includegraphics[width=3in]{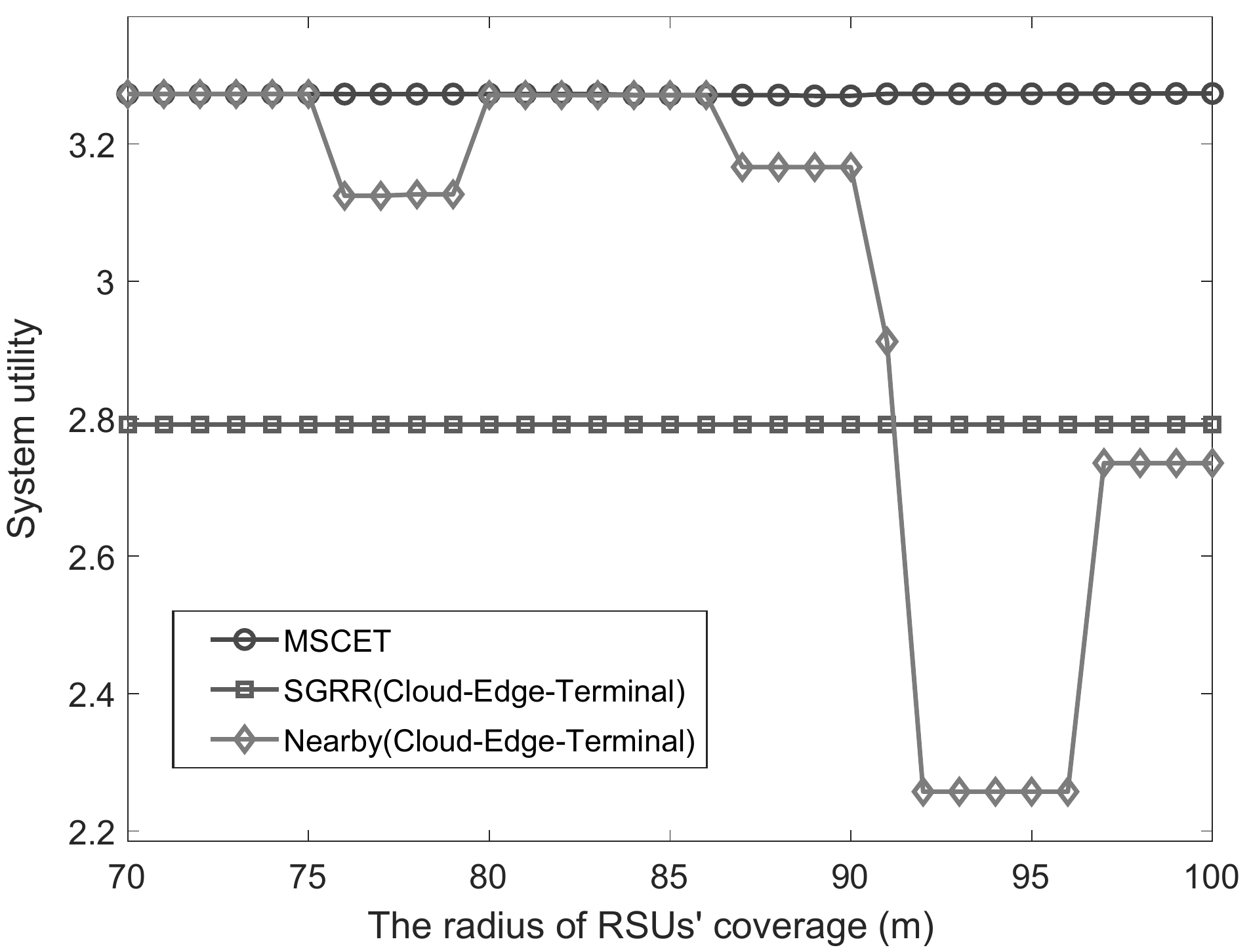}
  \caption{Performance comparison of system utility under different sizes of RSUs' coverage.}\label{RSUsize}
\end{figure}

\begin{figure}
  \centering
  \includegraphics[width=2.9in]{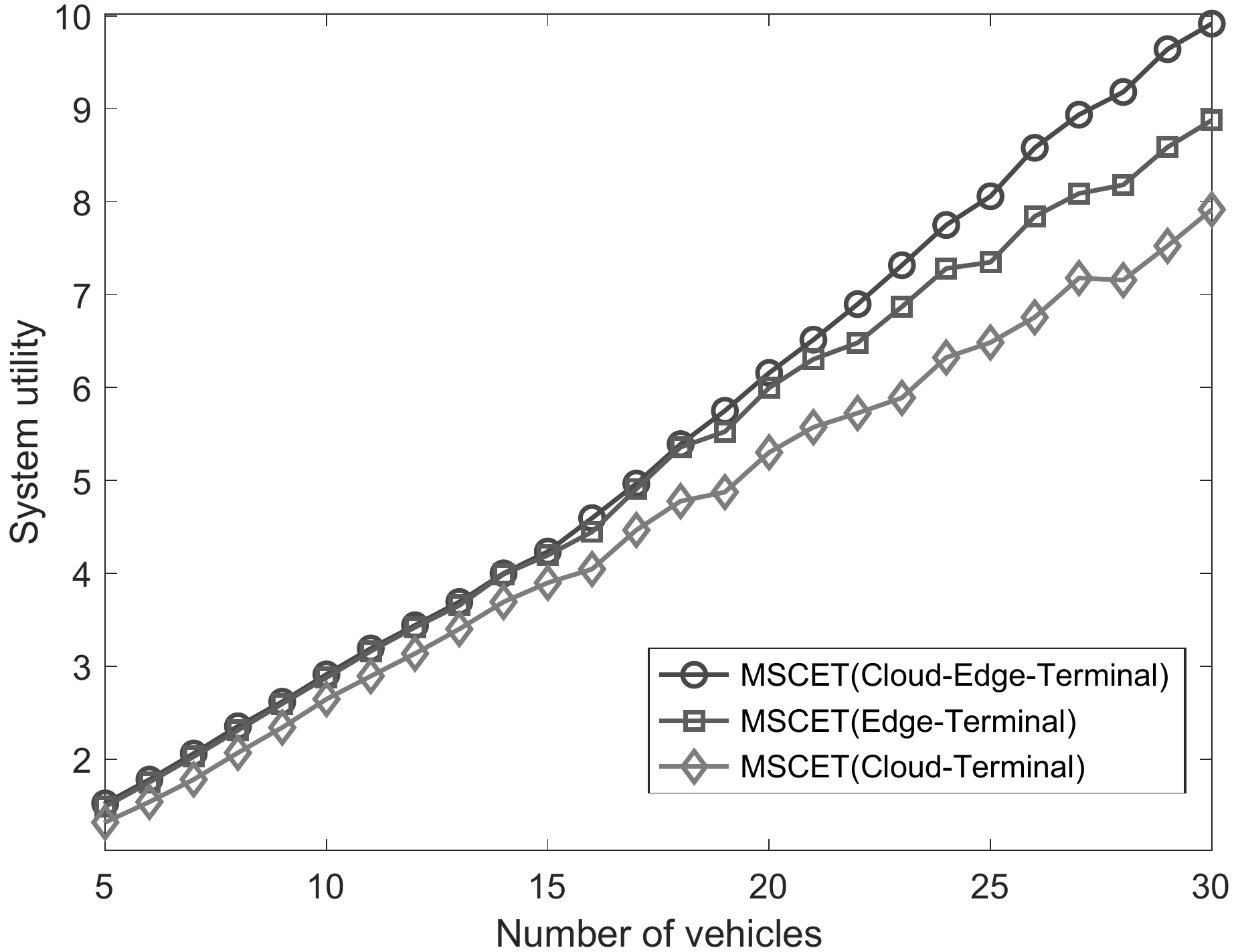}
  \caption{Performance comparison of system utility under different cooperation modes with the increase of vehicles in general region.}\label{models}
\end{figure}

\begin{figure*}
\centering \subfigure[Resource-rich]{
\includegraphics[width=1.98in]{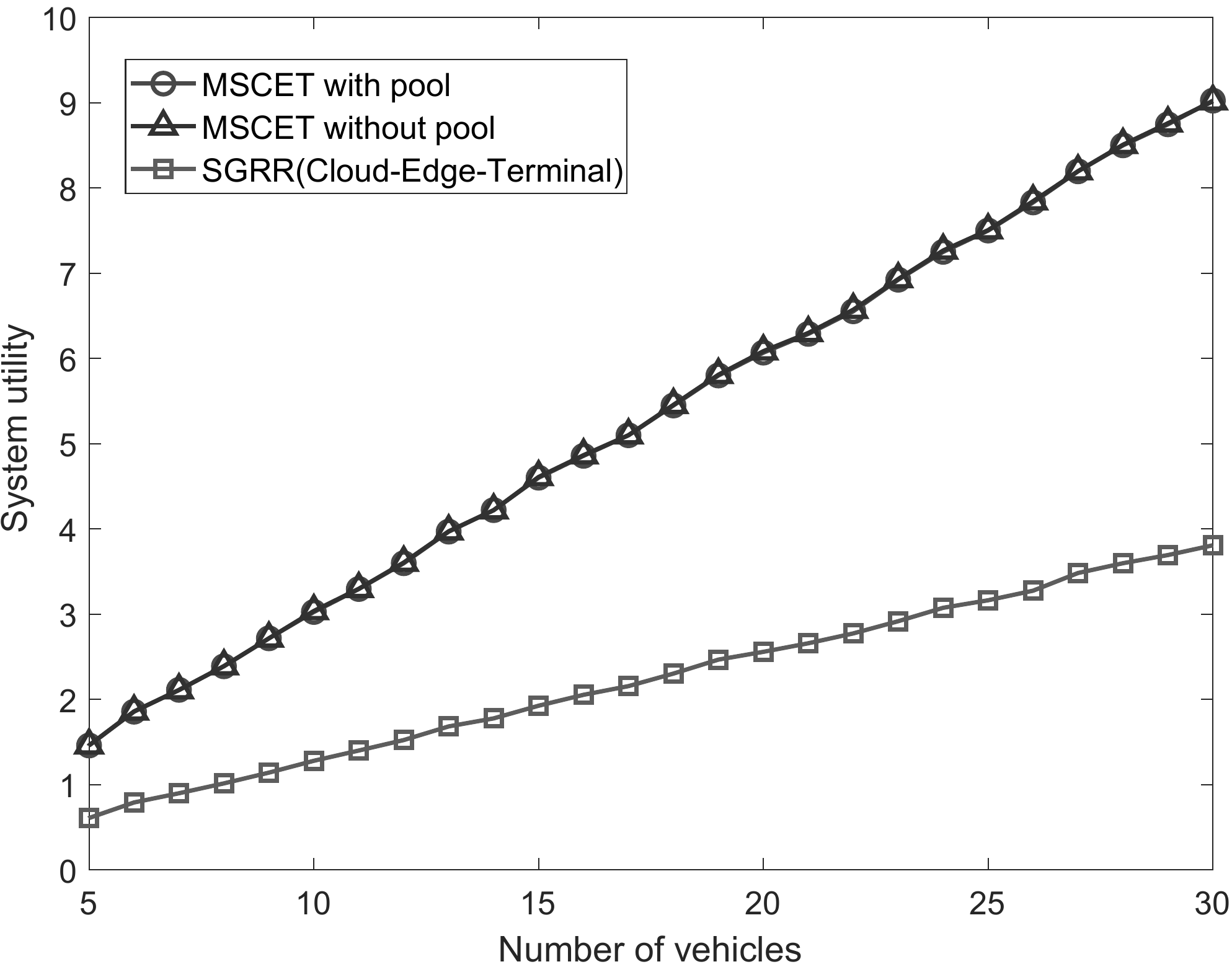}}
\subfigure[One resource-limited ES]{
\includegraphics[width=1.95in]{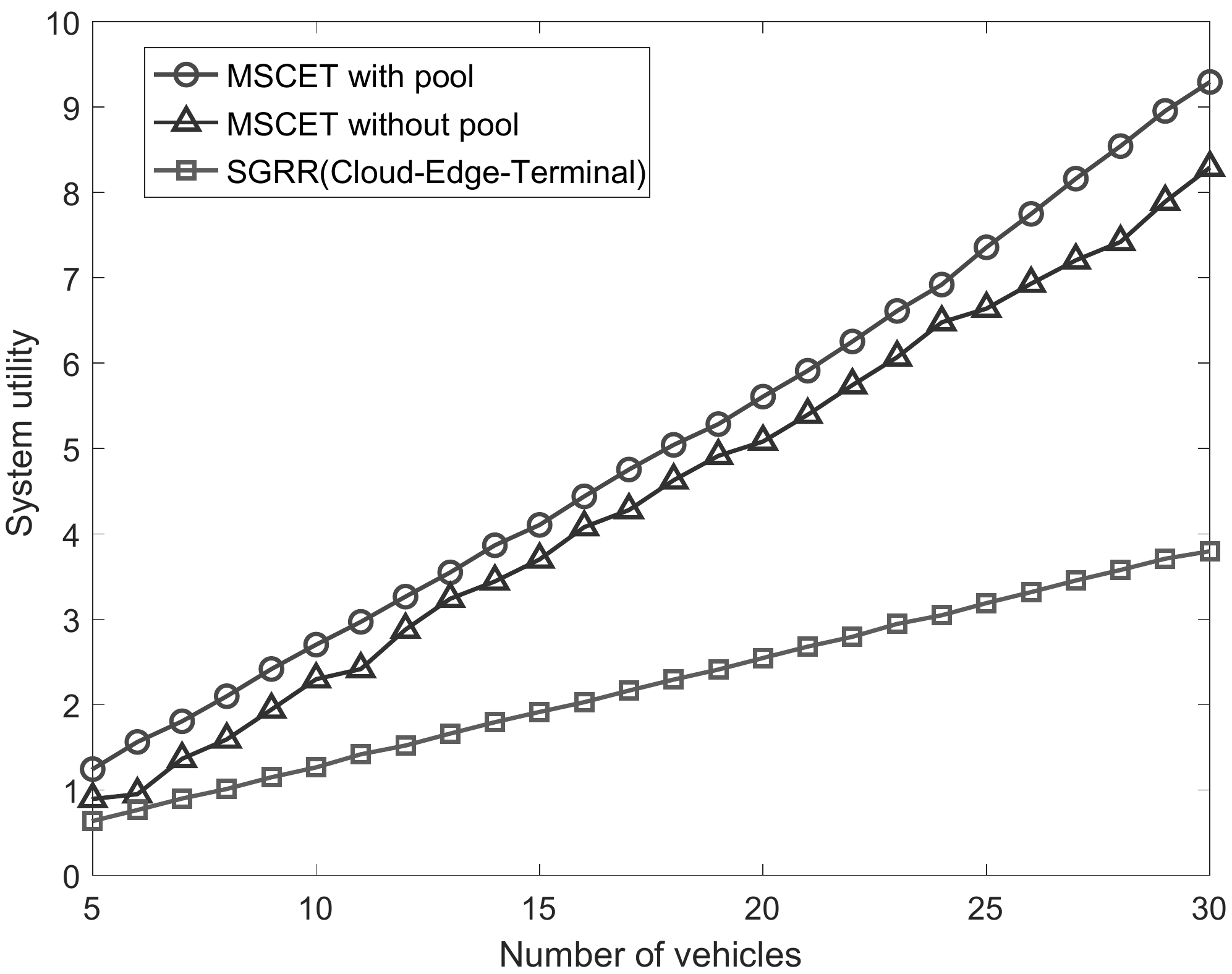}}
\subfigure[Three resource-limited ESs]{
\includegraphics[width=1.95in]{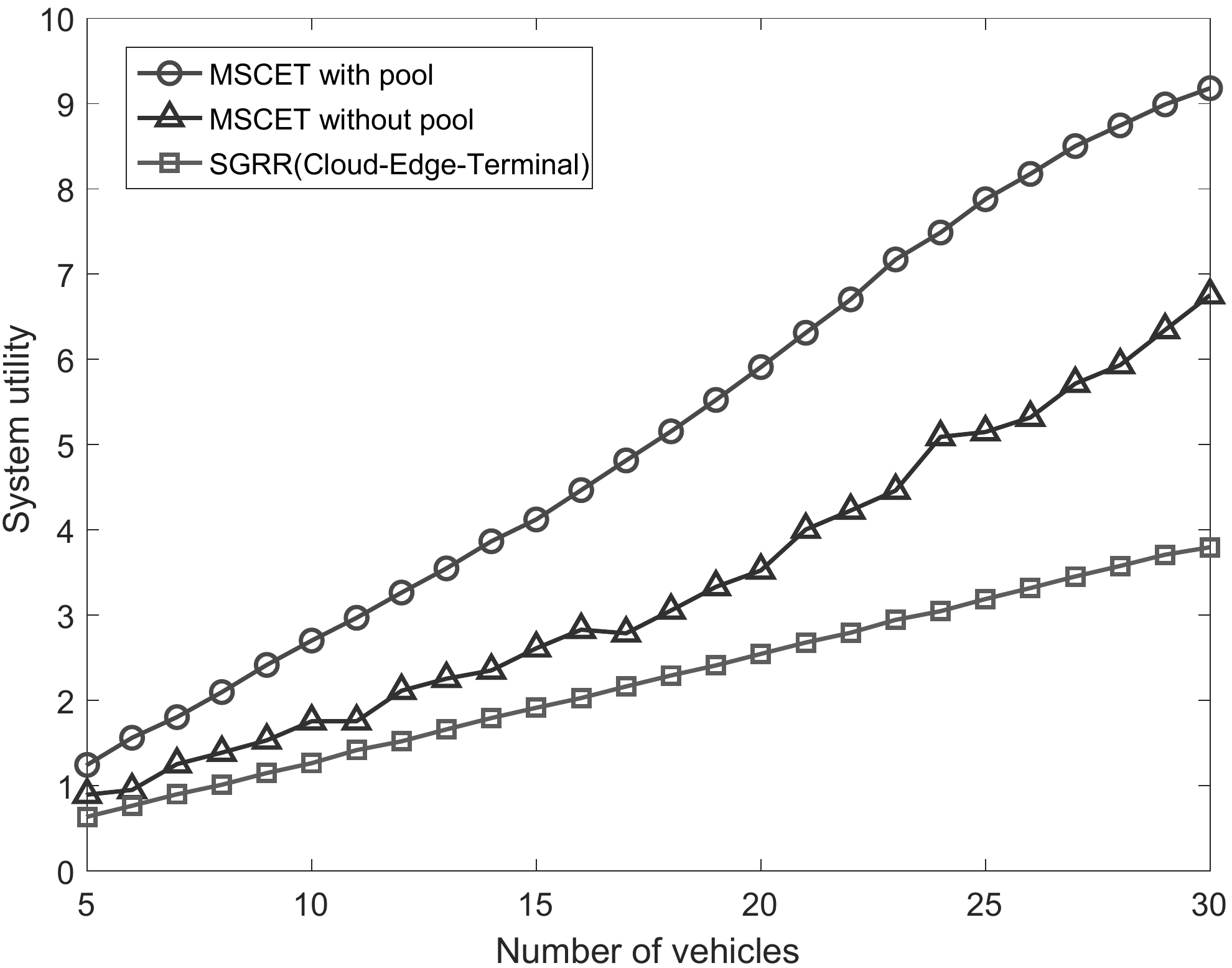}}
 \caption{{Performance comparison of system utility under different amount of resource with the increase of vehicles in overlapping region.}}\label{systemutilityO}
\end{figure*}

\subsection{Advantages of Cloud-Edge-Terminal Cooperation}

In this subsection, we compare the performance of the proposed MSCET with different cooperation modes in general region. Fig.~\ref{models} shows the impact of the increasing number of vehicles on the system utility between different cooperation modes. The results illustrate that with the increasing number of vehicles, the utility under all the three modes will increase. This is because that the increasing number of vehicles will produce more tasks that need to be processed. And Cloud-Terminal cooperation performs a bit worse than other modes due to the long communication delay between CS and vehicles. Furthermore, the performance of MSCET is similar to that of Edge-Terminal when there are fewer vehicles. While MSCET outperforms Edge-Terminal when the number of vehicles exceeds a certain threshold. The reason behind is that the proposed MSCET is a collaborative offloading schedule combining ESs and CS and the tasks can be processed locally and in servers (i.e., ESs and CS) simultaneously. In other words, the tasks can be reasonably offloaded to reduce the impact of the resource-limited ESs. Moreover, the proposed MSCET can make full use of the moving time to process the tasks by utilizing Cloud-Terminal cooperation while Edge-Terminal can only utilize local resource to process the tasks until reaching the coverage of the selected RSUs.

\subsection{Advantages of virtual resource pool}

To illustrate the benefits of virtual resource pool, we compare the performance of MSCET with and without the virtual resource pool under different amount of resource. Fig.~\ref{systemutilityO} illustrates that the performance of MSCET with or without the virtual resource pool is similar when the resource in overlapping region is abundant. However, the performance of MSCET without the pool will be worse when the resource is limited and the performance difference between MSCET with and without the pool increases with the decrease of resource. As the virtual resource pool integrates the resource from multiple distributed ESs. When some ESs are resource-limited, the pool can invoke the resource from other ESs to cooperate on the tasks. In contrast, ESs that without the assistance of the pool can only utilize their own resource to process the tasks and can be easily affected by the amount of the resource.

\section{Conclusions and Future Work}\label{sec6}

In this paper, we study the task offloading problem for vehicular networks and propose MSCET, a multi-scenario offloading schedule for biomedical data processing and analysis tasks in Cloud-Edge-Terminal collaborative vehicular networks, which integrates the Cloud-Edge-Terminal cooperation and virtual resource pool to improve the efficiency and reliability. We consider both the profits and QoS as the system utility, and formulate the system utility maximization problem. We then show that the problem is NP-hard and present an algorithm consisting of the Kuhn-Munkres matching method, genetic algorithm, interior point method and KKT condition to obtain efficient solutions. The parameters of the proposed MSCET are optimized to maximize the system utility. Finally, we conduct extensive simulations to evaluate the performance of the proposed MSCET and the simulation results illustrate that the MSCET outperforms existing schedules. However, some assumptions in our proposed MSCET are based on ideal condition and there may be some differences with the real scenario. In the future, we intend to extend the MSCET schedule to be applicable for the real-world scenario with implementation.

\section*{ACKNOWLEDGEMENT}
This work was supported in part by NSFC grants (No.61772551, No.62111530052),  and the Major Scientific and Technological Projects
of CNPC under Grant ZD2019-183-003.



%
\bibliographystyle{IEEEtranS}

\bibliography{IEEEabrv,IEEEtran}

%

\begin{IEEEbiography}[{\includegraphics[width=1in,height=1.25in,clip,keepaspectratio]{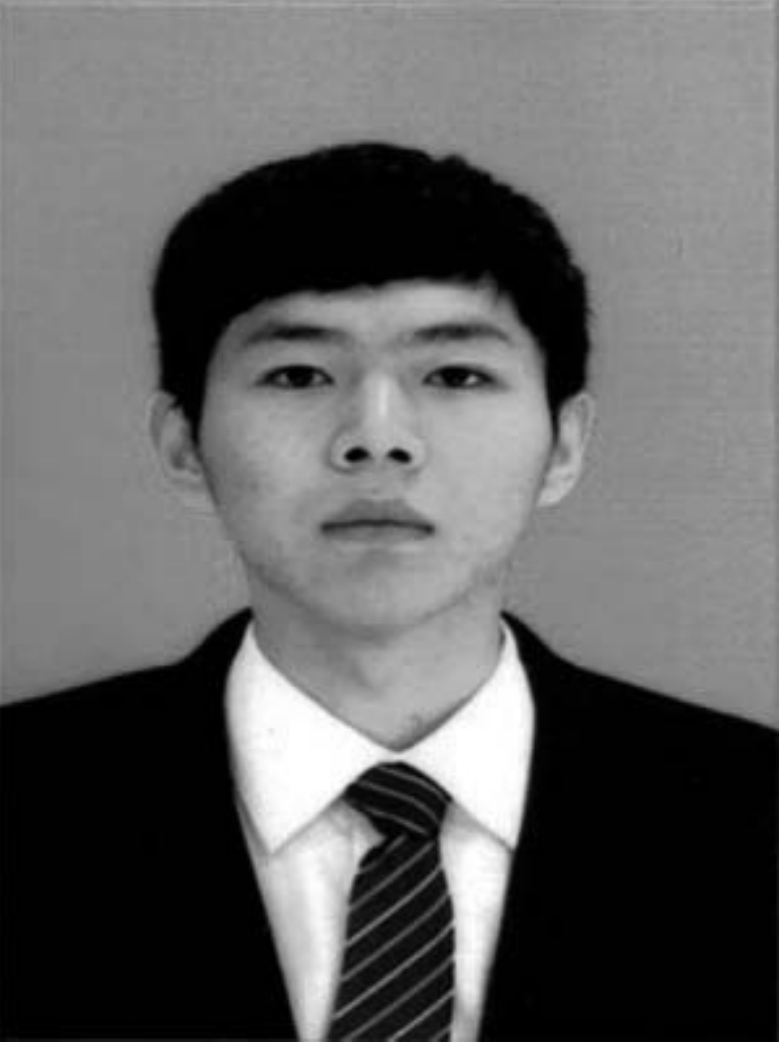}}]{Zhichen Ni}
  received the B.E. degree in measurement and control technology and instrumentation from China University of Petroleum, China, in 2019. He is currently pursuing his master degree in control science and engineering in the College of Control Science and Engineering, China University of Petroleum, China. His current research interests include edge computing and edge intelligence.
  \end{IEEEbiography}

  \begin{IEEEbiography}[{\includegraphics[width=1in,height=1.25in,clip,keepaspectratio]{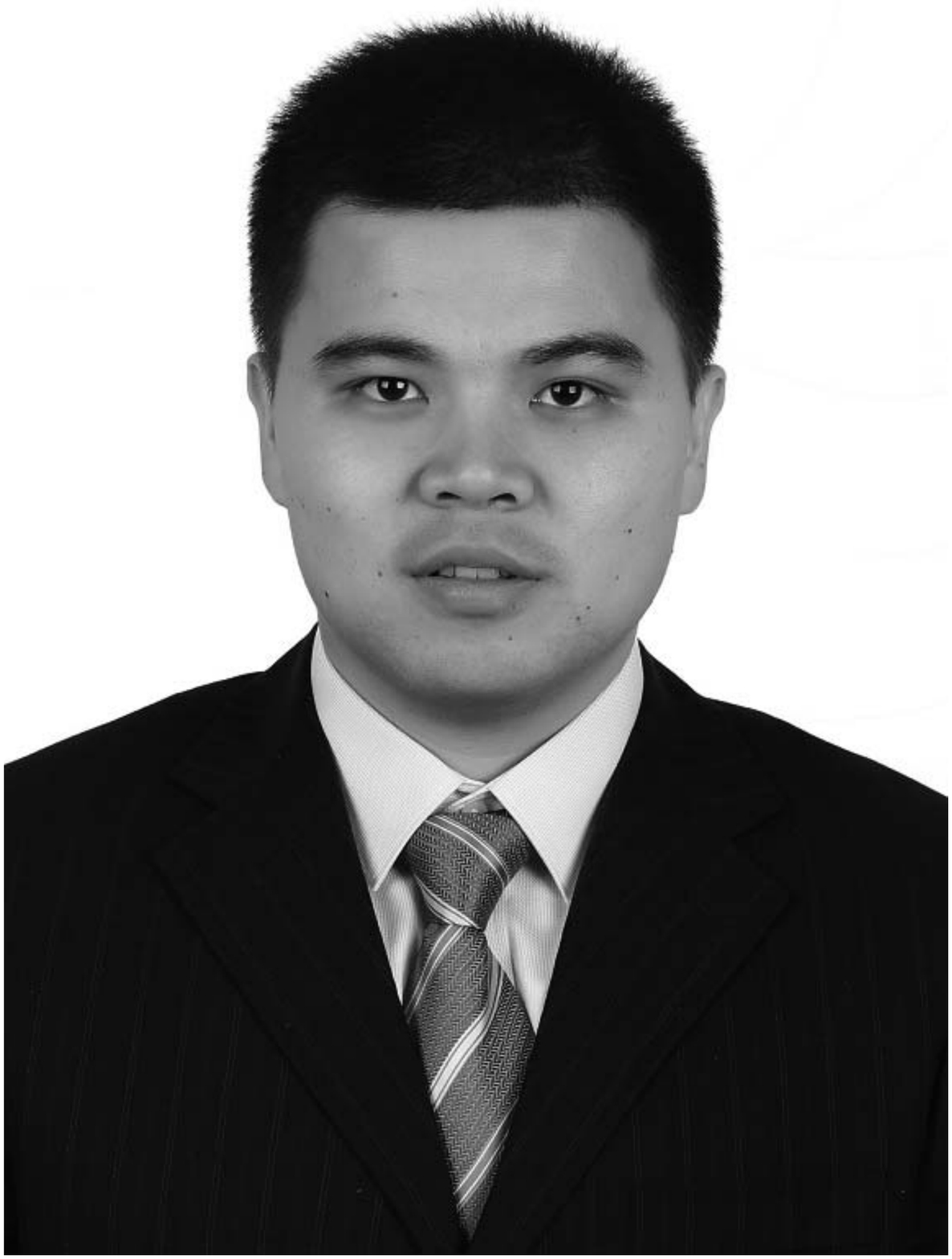}}]{Honglong Chen}
  received the M.E. degree in control science and
  engineering from Zhejiang University, China, in 2008, and
  the Ph.D degree in computer science from The Hong Kong Polytechnic
  University, Hong Kong, in 2012. He was a Postdoctoral Researcher in
  the School of CIDSE at Arizona State University from 2015
  to 2016. He is currently a Professor and Ph.D supervisor with the
  College of Control Science and Engineering, China University of
  Petroleum, China. His current research interests are in the areas of
  Internet of Things, edge computing and crowdsensing. He has published more than 80 research papers in prestigious journals and conferences including IEEE TIFS, IEEE TMC, IEEE TWC, IEEE TVT, IEEE IoT-J, IEEE INFOCOM, IEEE ICPP, IEEE ICDCS, etc. He is a senior member of IEEE and CCF (China Computer Federation), and a member of ACM.
  \end{IEEEbiography}

  \begin{IEEEbiography}[{\includegraphics[width=1in,height=1.25in]{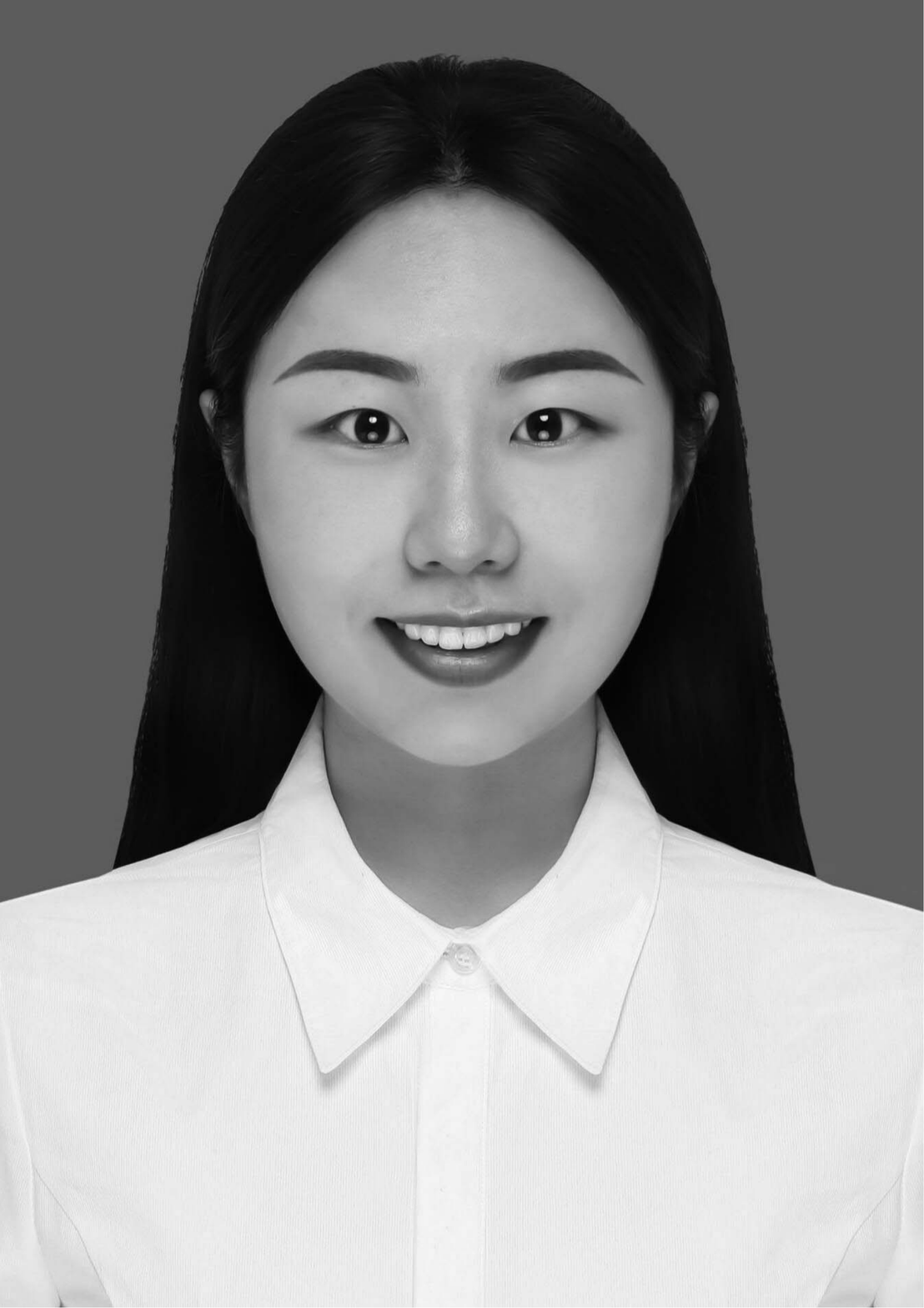}}]{Zhe Li}
  received the B. E. degree in automation from Shandong University of Science and Technology, China, in 2018. She is currently a graduate student in the College of Control Science and Engineering, China University of Petroleum, China. Her research interest is in the field of edge intelligence.
  \end{IEEEbiography}

  \begin{IEEEbiography}[{\includegraphics[width=1in,height=1.25in,clip,keepaspectratio]{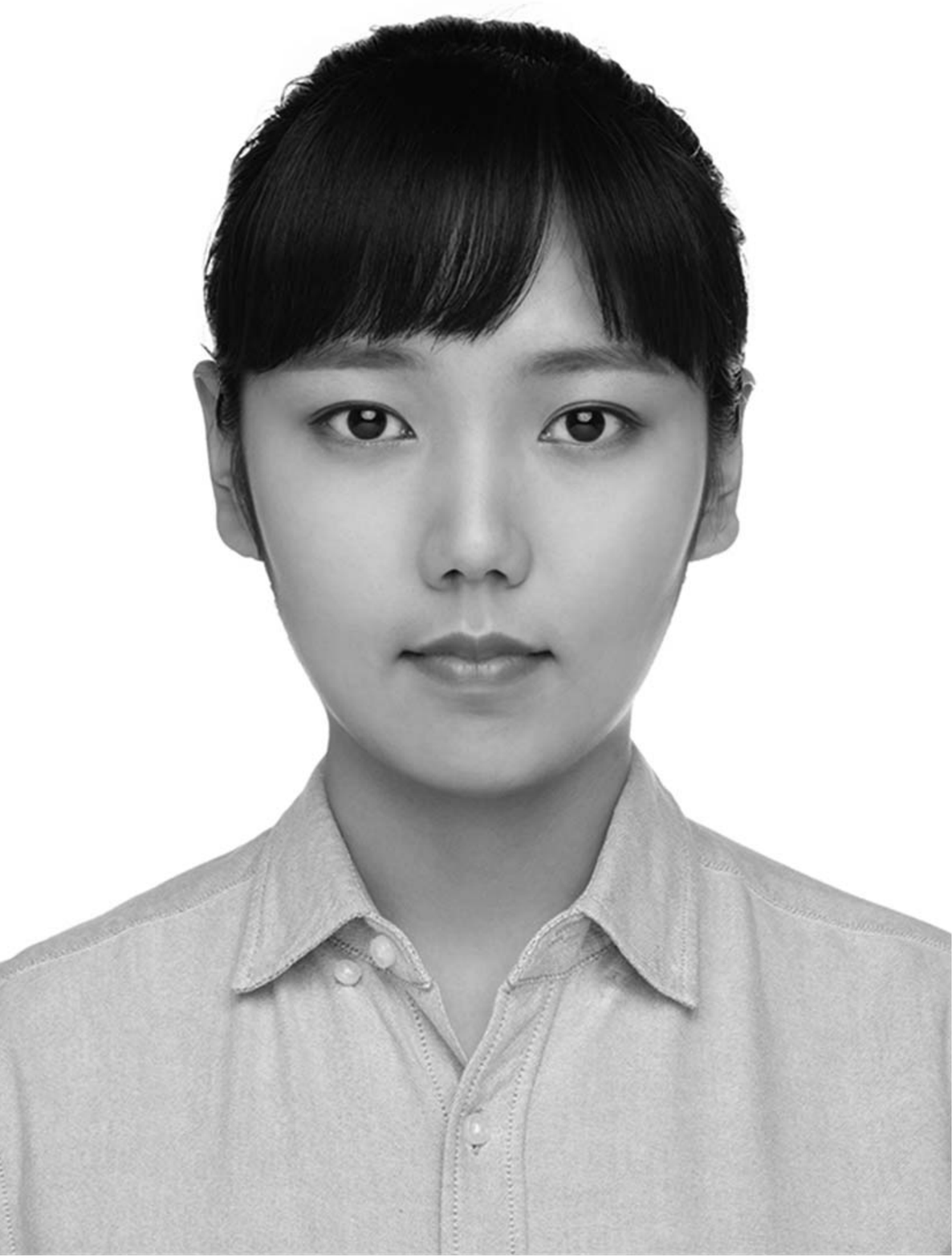}}]{Xiaomeng Wang}
  received the B.E. degree in measurement and control technology and instrumentation from China University of Petroleum, China, in 2020. She is currently pursuing her master degree in control science and engineering in the College of Control Science and Engineering, China University of Petroleum, China. Her current research interests include cybersecurity and artificial intelligence security.
  \end{IEEEbiography}

  \begin{IEEEbiography}[{\includegraphics[width=1in,height=1.25in,clip,keepaspectratio]{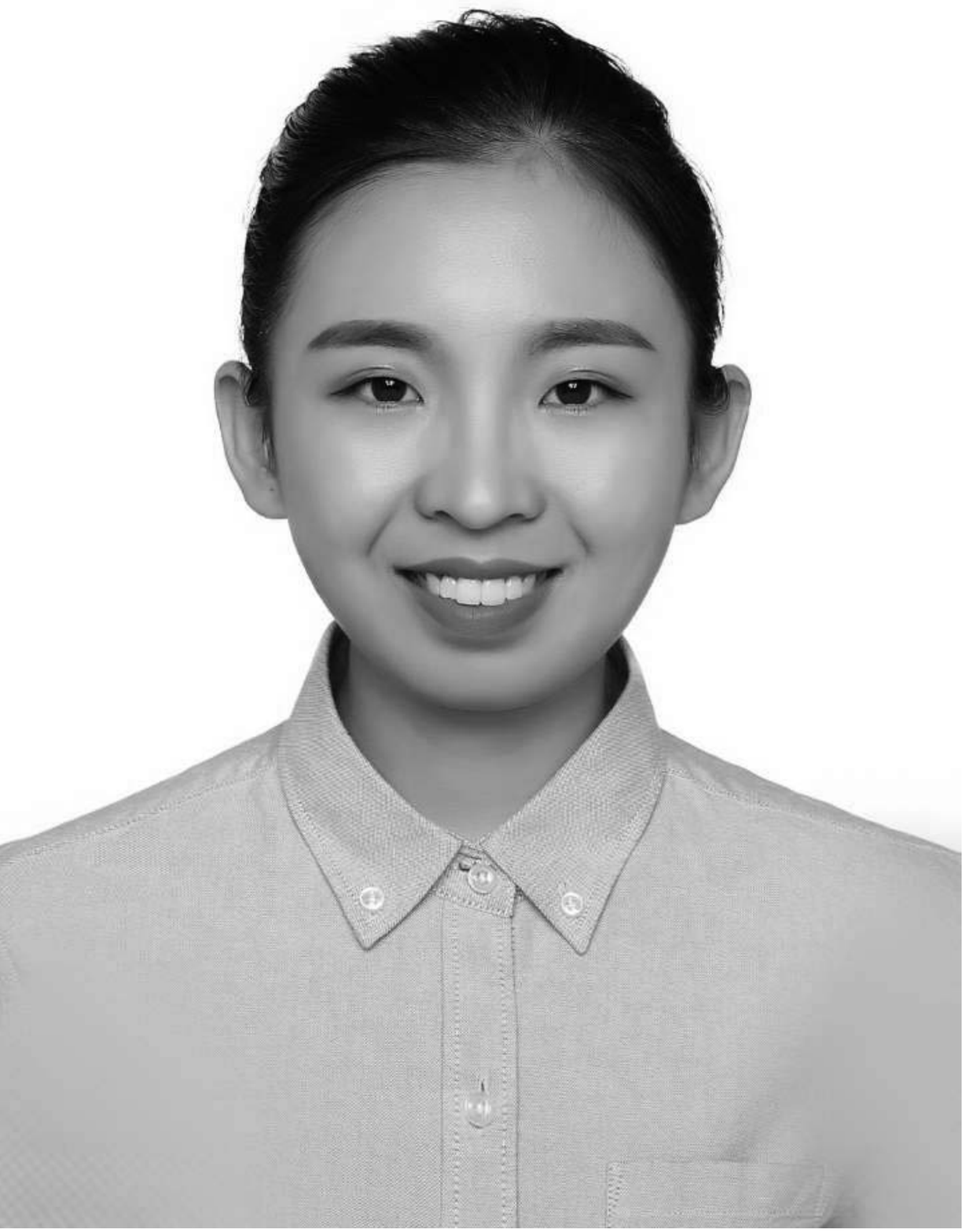}}]{Na Yan}
  received the B. E. degree in Automation from China University of Petroleum, China, in 2019. She is currently pursuing her master degree in control science and engineering in the College of Control Science and Engineering, China University of Petroleum, China. Her research interest is in the field of RFID.
  \end{IEEEbiography}

  \begin{IEEEbiography}[{\includegraphics[width=1in,height=1.25in,clip,keepaspectratio]{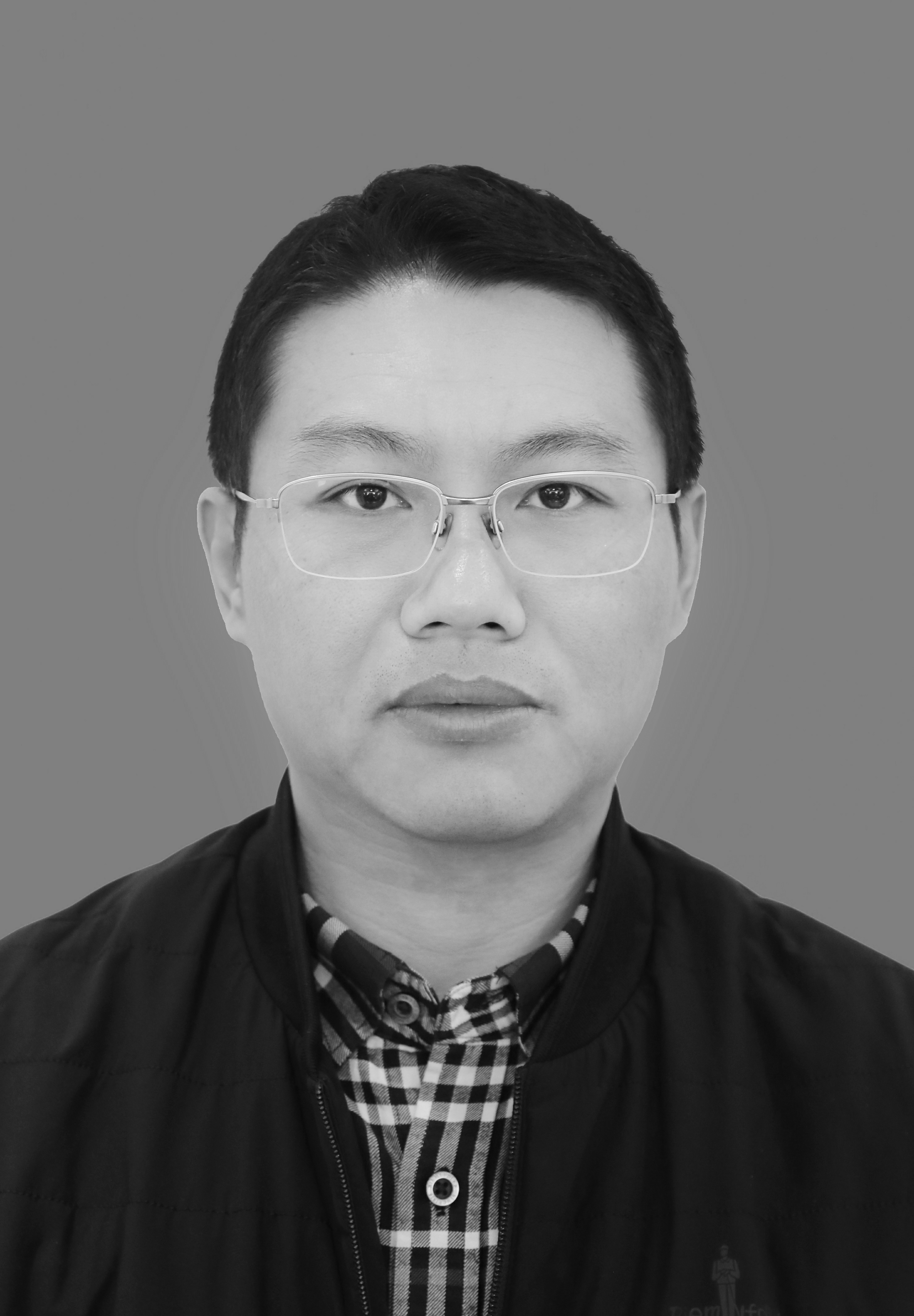}}]{Weifeng Liu}
   is currently a Professor with the College of Control Science and Engineering, China University of Petroleum (East China), China. He received the double B.S. degree in automation and business administration and the Ph.D. degree in pattern recognition and intelligent systems from the University of Science and Technology of China, Hefei, China, in 2002 and 2007, respectively. His current research interests include pattern recognition and machine learning. He has authored or co-authored a dozen papers in top journals and prestigious conferences including 10 ESI Highly Cited Papers and 3 ESI Hot Papers. Dr. Weifeng Liu serves as associate editor for Neural Processing Letter, co-chair for IEEE SMC technical committee on cognitive computing, and guest editor of special issue for Signal Processing, IET Computer Vision, Neurocomputing, and Remote Sensing. He also serves dozens of journals and conferences.
  \end{IEEEbiography}

  \begin{IEEEbiography}[{\includegraphics[width=1in,height=1.25in,clip,keepaspectratio]{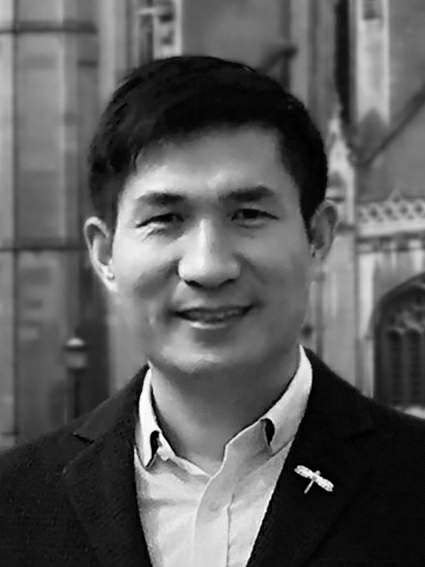}}]{Feng Xia}
  received the B.Sc. and Ph.D. degrees from Zhejiang University, Hangzhou, China. He is currently an Associate Professor and the Discipline Leader of the School of Science, Engineering and Information Technology, Federation University Australia, Ballarat, VIC, Australia. He has published two books and over 300 scientific papers in international journals and conferences. His research interests include data science, knowledge management, social computing, and systems engineering. He is a Senior Member of IEEE and ACM.
  \end{IEEEbiography}

  \end{document}